\documentclass[fleqn]{llncs}
\usepackage{version}
\usepackage{nuc}

\title{Instruction Sequence Based \\ Non-uniform Complexity Classes}
\author{J.A. Bergstra \and C.A. Middelburg}
\institute{Informatics Institute, Faculty of Science, University of
           Amsterdam, \\
           Science Park~904, 1098~XH Amsterdam, the Netherlands \\
           \email{J.A.Bergstra@uva.nl,C.A.Middelburg@uva.nl}}

\begin{document}
\maketitle

\begin{abstract}
We present an approach to non-uniform complexity in which single-pass 
instruction sequences play a key part, and answer various questions that
arise from this approach.
We introduce several kinds of non-uniform complexity classes.  
One kind includes a counterpart of the well-known non-uniform complexity 
class \PTpoly\ and another kind includes a counterpart of the well-known 
non-uniform complexity class \NPTpoly.
Moreover, we introduce a general notion of completeness for the 
non-uniform complexity classes of the latter kind.
We also formulate a counterpart of the well-known complexity theoretic 
conjecture that $\NPT \not\subseteq \PTpoly$.
We think that the presented approach opens up an additional way of
investigating issues concerning non-uniform complexity.
\begin{classcode}
F.1.1, F.1.3.
\end{classcode}
\end{abstract}

\section{Introduction}
\label{sect-intro}

The aim of this paper is to draw attention to an approach to non-uniform 
complexity which is based on the simple idea that each Boolean function 
can be computed by a single-pass instruction sequence that contains only 
instructions to read and write the contents of Boolean registers, 
forward jump instructions, and a termination instruction.

In the first place, we introduce a kind of non-uniform complexity 
classes which includes a counterpart of the classical non-uniform
complexity class \PTpoly\ and formulate a counterpart of the well-known
complexity theoretic conjecture that $\NPT \not\subseteq \PTpoly$.
Some evidence for this conjecture is the Karp-Lipton 
theorem~\cite{KL80a}.
The counterpart of the conjecture formulated in this paper is called the 
non-uniform super-polynomial complexity conjecture.
The counterpart of \PTpoly\ is denoted by \PLIS.

Over and above that, we introduce a kind of non-uniform complexity 
classes which includes a counterpart of the non-uniform complexity class 
\NPTpoly\ and introduce a general notion of completeness for the 
complexity classes of this kind. 
This general notion of completeness is defined using reducibility 
relations that can be regarded as non-uniform variants of the 
\pagebreak[2]
reducibility relation in terms of which \NPT-completeness is usually 
defined. 
The counterpart of \NPTpoly\ is denoted by \NPLIS.

We show among other things that the complexity classes \PTpoly\ and 
\NPTpoly\ coincide with the complexity classes \PLIS\ and \NPLIS, 
respectively, and that a problem closely related to \iiiSAT, and used to 
formulate the counterpart of the conjecture that 
$\NPT \not\subseteq \PTpoly$, is \NPT-complete and \NPLIS-complete.

In computer science, the meaning of programs usually plays a prominent
part in the explanation of many issues concerning programs.
Moreover, what is taken for the meaning of programs is mathematical by
nature.
Yet, it is customary that practitioners do not fall back on the 
mathematical meaning of programs in case explanation of issues 
concerning programs is needed.
They phrase their explanations from an empirical perspective.
An empirical perspective that we consider appealing is the perspective
that a program is in essence an instruction sequence and an instruction
sequence under execution produces a behaviour that is controlled by its
execution environment in the sense that each step of the produced
behaviour actuates the processing of an instruction by the execution 
environment and a reply returned at completion of the processing 
determines how the behaviour proceeds.

An attempt to approach the semantics of programming languages from
the perspective mentioned above is made in~\cite{BL02a}.
The groundwork for the approach is an algebraic theory of single-pass
instruction sequences, called program algebra, and an algebraic theory
of mathematical objects that represent the behaviours produced by 
instruction sequences under execution, called basic thread algebra.
The main advantages of the approach are that it does not require a lot
of mathematical background and that it is more appealing to
practitioners than the main approaches to programming language
semantics.

As a continuation of the work on the above-mentioned approach to
programming language semantics, the notion of an instruction sequence
was subjected to systematic and precise analysis using the groundwork
laid earlier.
This led among other things to expressiveness results about the
instruction sequences considered and variations of the instruction
sequences considered (see e.g.~\cite{BM08h,BP09a}). \linebreak[2]
As another continuation of the work on the above-mentioned approach to
programming language semantics, selected issues relating to well-known
subjects from the theory of computation and the area of computer
architecture were rigorously investigated thinking in terms of
instruction sequences (see e.g.~\cite{BM09i,BM09k}).
The general aim of the work in both continuations mentioned is to bring 
instruction sequences as a theme in computer science better into the 
picture. 
The work presented in this paper forms a part of the last mentioned 
continuation.

The starting-point of program algebra is the perception of a program as
a single-pass instruction sequence, i.e.\ a finite or infinite sequence
of instructions of which each instruction is executed at most once and
can be dropped after it has been executed or jumped over.
This perception is simple, appealing, and links up with practice.
The concepts underlying the primitives of program algebra are common in
programming, but the particular form of the primitives is not common.
The predominant concern in the design of program algebra has been to
achieve simple syntax and semantics, while maintaining the expressive
power of arbitrary finite control.

The objects considered in basic thread algebra represent in a direct way 
the behaviours produced by instruction sequences under execution: upon 
each action performed by such an object, a reply from an execution 
environment, which takes the action as an instruction to be processed, 
determines how it proceeds.
The objects concerned are called threads.
A thread may make use of services, i.e.\ components of the execution
environment.
Once introduced into threads and services, it is rather obvious that
each Turing machine can be simulated by means of a thread that makes use
of a service.
The thread and service correspond to the finite control and tape of the
Turing machine.

The approach to complexity followed in this paper is not suited to
uniform complexity.
This is not considered a great drawback.
Non-uniform complexity is the relevant notion of complexity when
studying what looks to be the major complexity issue in practice: the
scale-dependence of what is an efficient solution for a computational
problem.

This paper is organized as follows.
First, we survey program algebra and basic thread algebra
(Section~\ref{sect-PGA-and-BTA}).
Next, we survey an extension of basic thread algebra concerning the 
interaction of threads with services and give a description of Boolean 
register services 
(Sections~\ref{sect-TSI} and~\ref{sect-Boolean-register}).
Then, we introduce the kind of complexity classes that includes \PLIS\ 
and formulate the non-uniform super-polynomial complexity conjecture
(Sections~\ref{sect-complclass-PLIS}, \ref{sect-boolform-and-boolcirc} 
and~\ref{sect-conjecture}).
After that, we introduce the kind of complexity classes that includes 
\NPLIS\ and the notion of completeness for the non-uniform complexity 
classes of this kind
(Sections~\ref{sect-complclass-NPLIS} and~\ref{sect-NPLIS-compl}).
We also introduce two additional kinds of complexity classes suggested 
by the problem closely related to \iiiSAT\ that is used earlier
(Section~\ref{sect-projective}).
Finally, we make some concluding remarks (Section~\ref{sect-concl}).

Some familiarity with classical computational complexity is assumed.
The relevant notions are explained in many textbooks, 
including~\cite{AB09a,BDG88a,Gol08a}.
Their precise definitions in different publications differ slightly.
The definitions of classical notions on which some results in this paper
are based are the ones from Chapters~1, 2 and~6 of~\cite{AB09a}.

This paper supersedes~\cite{BM08g} and Section~5.2 of~\cite{BM12b} in 
several respects.
Generalization of the definitions of the complexity classes \PLIS\ and
\NPLIS\ has put these complexity classes into a broader context, and a 
major technical change has made it possible to simplify the material 
that is concerned with the complexity class \NPLIS.
Moreover, two additional kinds of complexity classes are introduced, and 
various additional results are given.

\section{Program Algebra and Basic Thread Algebra}
\label{sect-PGA-and-BTA}

In this section, we survey \PGA\ (ProGram Algebra) and \BTA\ (Basic 
Thread Algebra) and make precise in the setting of \BTA\ which 
behaviours are produced on execution by the instruction sequences 
considered in \PGA.

In \PGA, it is assumed that there is a fixed but arbitrary set $\BInstr$
of \emph{basic instructions}.
The intuition is that the execution of a basic instruction may modify a 
state and produces a reply at its completion.
The possible replies are $\True$ and $\False$.
The actual reply is generally state-dependent.
Therefore, successive executions of the same basic instruction may
produce different replies.
The set $\BInstr$ is the basis for the set of instructions that may 
occur in the instruction sequences considered in \PGA.
The elements of the latter set are called \emph{primitive instructions}.

\PGA\ has the following primitive instructions:
\begin{itemize}
\item
for each $a \in \BInstr$, a \emph{plain basic instruction} $a$;
\item
for each $a \in \BInstr$, a \emph{positive test instruction} $\ptst{a}$;
\item
for each $a \in \BInstr$, a \emph{negative test instruction} $\ntst{a}$;
\item
for each $l \in \Nat$, a \emph{forward jump instruction} $\fjmp{l}$;
\item
a \emph{termination instruction} $\halt$.
\end{itemize}
We write $\PInstr$ for the set of all primitive instructions.

On execution of an instruction sequence, these primitive instructions
have the following effects:
\begin{itemize}
\item
the effect of a positive test instruction $\ptst{a}$ is that basic
instruction $a$ is executed and execution proceeds with the next
primitive instruction if $\True$ is produced and otherwise the next
primitive instruction is skipped and execution proceeds with the
primitive instruction following the skipped one --- if there is no
primitive instruction to proceed with,
inaction occurs;
\item
the effect of a negative test instruction $\ntst{a}$ is the same as
the effect of $\ptst{a}$, but with the role of the value produced
reversed;
\item
the effect of a plain basic instruction $a$ is the same as the effect
of $\ptst{a}$, \linebreak[2] but execution always proceeds as if $\True$ 
is produced;
\item
the effect of a forward jump instruction $\fjmp{l}$ is that execution
proceeds with the $l$th next primitive instruction of the instruction
sequence concerned --- if $l$ equals $0$ or there is no primitive
instruction to proceed with, inaction occurs;
\item
the effect of the termination instruction $\halt$ is that execution 
terminates.
\end{itemize}

\PGA\ has one sort: the sort $\InSeq$ of \emph{instruction sequences}. 
We make this sort explicit to anticipate the need for many-sortedness
later on.
To build terms of sort $\InSeq$, \PGA\ has the following constants and 
operators:
\begin{itemize}
\item
for each $u \in \PInstr$, 
the \emph{instruction} constant $\const{u}{\InSeq}$\,;
\item
the binary \emph{concatenation} operator 
$\funct{\ph \conc \ph}{\InSeq \x \InSeq}{\InSeq}$\,;
\item
the unary \emph{repetition} operator 
$\funct{\ph\rep}{\InSeq}{\InSeq}$\,.
\end{itemize}
Terms of sort $\InSeq$ are built as usual.
Throughout the paper, we assume that there are infinitely many variables 
of sort $\InSeq$, including $X,Y,Z$.
We use infix notation for concatenation and postfix notation for
repetition.

A closed \PGA\ term is considered to denote a non-empty, finite or
eventually periodic infinite sequence of primitive instructions.%
\footnote
{An eventually periodic infinite sequence is an infinite sequence with
 only finitely many distinct suffixes.}
The instruction sequence denoted by a closed term of the form
$P \conc Q$ is the instruction sequence denoted by $P$
concatenated with the instruction sequence denoted by $Q$.
The instruction sequence denoted by a closed term of the form $P\rep$
is the instruction sequence denoted by $P$ concatenated infinitely
many times with itself.

Closed \PGA\ terms are considered equal if they represent the same
instruction sequence.
The axioms for instruction sequence equivalence are given in
Table~\ref{axioms-PGA}.%
\begin{table}[!tb]
\caption{Axioms of \PGA}
\label{axioms-PGA}
\begin{eqntbl}
\begin{axcol}
(X \conc Y) \conc Z = X \conc (Y \conc Z)              & \axiom{PGA1} \\
(X^n)\rep = X\rep                                      & \axiom{PGA2} \\
X\rep \conc Y = X\rep                                  & \axiom{PGA3} \\
(X \conc Y)\rep = X \conc (Y \conc X)\rep              & \axiom{PGA4}
\end{axcol}
\end{eqntbl}
\end{table}
In this table, $n$ stands for an arbitrary natural number greater than
$0$.
For each $n > 0$, the term $P^n$, where $P$ is a \PGA\ term, is defined 
by induction on $n$ as follows: $P^1 = P$ and $P^{n+1} = P \conc P^n$.
The \emph{unfolding} equation $X\rep = X \conc X\rep$ is
derivable.
Each closed \PGA\ term is derivably equal to a term in
\emph{canonical form}, i.e.\ a term of the form $P$ or $P \conc Q\rep$,
where $P$ and $Q$ are closed \PGA\ terms in which the repetition
operator does not occur.

A typical model of \PGA\ is the model in which:
\begin{itemize}
\item
the domain is the set of all finite and eventually periodic infinite
sequences over the set $\PInstr$ of primitive instructions;
\item
the operation associated with ${} \conc {}$ is concatenation;
\item
the operation associated with ${}\rep$ is the operation ${}\srep$
defined as follows:

\begin{itemize}
\item
if $U$ is a finite sequence, then $U\srep$ is the unique eventually
periodic infinite sequence $U'$ such that $U$ concatenated $n$ times 
with itself is a proper prefix of $U'$ for each $n \in \Nat$;
\item
if $U$ is an eventually periodic infinite sequence, then $U\srep$ is
$U$.
\end{itemize}
\end{itemize}
To simplify matters, we confine ourselves to this model of \PGA, which
is an initial model of \PGA, for the interpretation of \PGA\ terms.
In the sequel, we use the term \emph{instruction sequence} for the 
elements of the domain of this model, and we denote the interpretations
of the constants and operators in this model by the constants and 
operators themselves.

In the remainder of this paper, we consider instruction sequences that
can be denoted by closed \PGA\ terms in which the repetition operator
does not occur.
Below, we will make precise which behaviours are produced by
instruction sequences that can be denoted by closed \PGA\ terms in
which the repetition operator does not occur.

First, we survey \BTA, an algebraic theory of mathematical objects which 
represent in a direct way the behaviours produced by instruction 
sequences under execution.

In \BTA, it is assumed that a fixed but arbitrary set $\BAct$ of
\emph{basic actions}, with $\Tau \not\in \BAct$, has been given.
Besides, $\Tau$ is a special basic action.
We write $\BActTau$ for $\BAct \union \set{\Tau}$.

The objects considered in \BTA\ are called threads.
A thread represents a behaviour which consists of performing basic 
actions in a sequential fashion.
Upon each basic action performed, a reply from an execution environment
determines how the thread proceeds.
The possible replies are the Boolean values $\True$ and $\False$.
Performing $\Tau$, which is considered performing an internal action,
will always lead to the reply $\True$.

\BTA\ has one sort: the sort $\Thr$ of \emph{threads}. 
We make this sort explicit to anticipate the need for many-sortedness
later on.
To build terms
of sort $\Thr$, \BTA\ has the following constants and operators:
\begin{itemize}
\item
the \emph{inaction} constant $\const{\DeadEnd}{\Thr}$;
\item
the \emph{termination} constant $\const{\Stop}{\Thr}$;
\item
for each $\alpha \in \BActTau$, the binary \emph{postconditional
composition} operator 
$\funct{\pcc{\ph}{\alpha}{\ph}}{\linebreak[2] \Thr \x \Thr}{\Thr}$.
\end{itemize}
Terms of sort $\Thr$ are built as usual. Throughout the paper, we
assume that there are infinitely many variables of sort $\Thr$,
including $x,y,z$.
We use infix notation for postconditional composition. 

We introduce \emph{basic action prefixing} as an abbreviation: 
$\alpha \bapf p$, where $p$ is a \BTA\ term, abbreviates 
$\pcc{p}{\alpha}{p}$.
We identify expressions of the form $\alpha \bapf p$ with the \BTA\
term they stand for.

The thread denoted by a closed term of the form $\pcc{p}{\alpha}{q}$
will first perform $\alpha$, and then proceed as the thread denoted by
$p$ if the reply from the execution environment is $\True$ and proceed
as the thread denoted by $q$ if the reply from the execution
environment is $\False$. The thread denoted by $\DeadEnd$ will become
inactive and the thread denoted by $\Stop$ will terminate.

\BTA\ has only one axiom.
This axiom is given in Table~\ref{axioms-BTA}.%
\begin{table}[!t]
\caption{Axiom of \BTA} \label{axioms-BTA}
\begin{eqntbl}
\begin{axcol}
\pcc{x}{\Tau}{y} = \pcc{x}{\Tau}{x}                      & \axiom{T1}
\end{axcol}
\end{eqntbl}
\end{table}
Using the abbreviation introduced above, axiom T1 can be written as
follows: $\pcc{x}{\Tau}{y} = \Tau \bapf x$.

Each closed \BTA\ term denotes a finite thread, i.e.\ a thread with a
finite upper bound to the number of basic actions that it can perform.
Infinite threads, i.e.\ threads without a finite upper bound to the
number of basic actions that it can perform, can be defined by means of 
a set of recursion equations (see e.g.~\cite{BM07a}).
Regular threads, i.e.\ finite or infinite threads that can only be in a 
finite number of states, can be defined by means of a finite set of 
recursion equations.

The behaviours of the instruction sequences denoted by closed \PGA\
terms are considered to be regular threads, with the basic instructions
taken for basic actions.
All regular threads in which $\Tau$ does not occur represent behaviours 
of instruction sequences that can be denoted by closed \PGA\ terms
(see Proposition~2 in~\cite{PZ06a}).
Closed \PGA\ terms in which the repetition operator does not occur
correspond to finite threads.

Henceforth, we will write \PGAfin\ for \PGA\ without the repetition
operator and axioms PGA2--PGA4, and we will write $\FIS$ for the set of
all instruction sequences that can be denoted by closed \PGAfin\ terms.
Moreover, we will write $\psize(U)$, where $U \in \FIS$, for the length
of $U$.

We combine \PGAfin\ with \BTA\ and extend the combination with
the \emph{thread extraction} operator $\funct{\extr{\ph}}{\InSeq}{\Thr}$ 
and the axioms given in Table~\ref{axioms-thread-extr}.%
\begin{table}[!tb]
\caption{Axioms for the thread extraction operator}
\label{axioms-thread-extr}
\begin{eqntbl}
\begin{eqncol}
\extr{a} = a \bapf \DeadEnd \\
\extr{a \conc X} = a \bapf \extr{X} \\
\extr{\ptst{a}} = a \bapf \DeadEnd \\
\extr{\ptst{a} \conc X} =
\pcc{\extr{X}}{a}{\extr{\fjmp{2} \conc X}} \\
\extr{\ntst{a}} = a \bapf \DeadEnd \\
\extr{\ntst{a} \conc X} =
\pcc{\extr{\fjmp{2} \conc X}}{a}{\extr{X}}
\end{eqncol}
\qquad
\begin{eqncol}
\extr{\fjmp{l}} = \DeadEnd \\
\extr{\fjmp{0} \conc X} = \DeadEnd \\
\extr{\fjmp{1} \conc X} = \extr{X} \\
\extr{\fjmp{l+2} \conc u} = \DeadEnd \\
\extr{\fjmp{l+2} \conc u \conc X} = \extr{\fjmp{l+1} \conc X} \\
\extr{\halt} = \Stop \\
\extr{\halt \conc X} = \Stop
\end{eqncol}
\end{eqntbl}
\end{table}
In this table, $a$ stands for an arbitrary basic instruction from 
$\BInstr$, $u$ stands for an arbitrary primitive instruction from 
$\PInstr$, and $l$ stands for an arbitrary natural number.

For each closed \PGAfin\ term $P$, $\extr{P}$ denotes the behaviour  
produced by the instruction sequence denoted by $P$ under execution.
The use of a closed \PGAfin\ term is sometimes preferable to the use of
the corresponding closed \BTA\ term because thread extraction can give
rise to a combinatorial explosion.
For instance, suppose that $p$ is a closed \BTA\ term such that
\begin{ldispl}
p = 
\extr{\overbrace{\ptst{a} \conc \ptst{b} \conc \ldots \conc
                 \ptst{a} \conc \ptst{b}}^{k\; \x} {} \conc
                 c \conc \halt}\;.
\end{ldispl}%
Then the size of $p$ is greater than $2^{k/2}$.

\section{Interaction of Threads with Services}
\label{sect-TSI}

A thread may perform a basic action for the purpose of requesting a
named service provided by an execution environment to process a method
and to return a reply to the thread at completion of the processing of
the method.
In this section, we survey the extension of \BTA\ with services and
operators that are concerned with this kind of interaction between
threads and services.

It is assumed that a fixed but arbitrary set $\Foci$ of \emph{foci} has 
been given.
Foci play the role of names of the services provided by an execution
environment.
It is also assumed that a fixed but arbitrary set $\Meth$ of 
\emph{methods} has been given.
For the set $\BAct$ of basic actions, we take the set
$\set{f.m \where f \in \Foci, m \in \Meth}$.
Performing a basic action $f.m$ is taken as making a request to the
service named $f$ to process method $m$.

A service is able to process certain methods.
The processing of a method may involve a change of the service.
The reply value produced by the service at completion of the processing
of a method is either $\True$, $\False$ or $\Blocked$.
The special reply $\Blocked$, standing for blocked, is used to deal with
the situation that a service is requested to process a method that it is
not able to process.

The following is assumed with respect to services:
\begin{itemize}
\item
a many-sorted signature $\Sig{\Services}$ has been given that includes 
the following sorts:
\begin{itemize}
\item
the sort $\Serv$ of \emph{services};
\item
the sort $\Repl$ of \emph{replies};
\end{itemize}
\pagebreak[2]
and the following constants and operators:
\begin{itemize}
\item
the
\emph{empty service} constant $\const{\emptyserv}{\Serv}$;
\item
the \emph{reply} constants $\const{\True,\False,\Blocked}{\Repl}$;
\item
for each $m \in \Meth$, the
\emph{derived service} operator $\funct{\derive{m}}{\Serv}{\Serv}$;
\item
for each $m \in \Meth$, the
\emph{service reply} operator $\funct{\sreply{m}}{\Serv}{\Repl}$;
\end{itemize}
\item
a minimal $\Sig{\Services}$-algebra $\ServAlg$ has been given in which
$\True$, $\False$, and $\Blocked$ are mutually different, and
\begin{itemize}
\item
$\LAND{m \in \Meth}
  \derive{m}(z) = z \Land \sreply{m}(z) = \Blocked \Limpl
  z = \emptyserv$
holds;
\item
for each $m \in \Meth$,
$\derive{m}(z) = \emptyserv \Liff \sreply{m}(z) = \Blocked$ holds.
\end{itemize}
\end{itemize}

The intuition concerning $\derive{m}$ and $\sreply{m}$ is that on a
request to service $S$ to process method $m$:
\begin{itemize}
\item
if $\sreply{m}(S) \neq \Blocked$, $S$ processes $m$, produces the reply
$\sreply{m}(S)$, and then proceeds as $\derive{m}(S)$;
\item
if $\sreply{m}(S) = \Blocked$, $S$ is not able to process method $m$ and
proceeds as $\emptyserv$.
\end{itemize}
The empty service $\emptyserv$ itself is unable to process any method.

We introduce the following additional operators:
\begin{itemize}
\item
for each $f \in \Foci$, the binary \emph{use} operator
$\funct{\use{\ph}{f}{\ph}}{\Thr \x \Serv}{\Thr}$;
\item
for each $f \in \Foci$, the binary \emph{apply} operator
$\funct{\apply{\ph}{f}{\ph}}{\Thr \x \Serv}{\Serv}$.
\end{itemize}
We use infix notation for the use and apply operators.

The thread denoted by a closed term of the form $\use{p}{f}{S}$ and the
service denoted by a closed term of the form $\apply{p}{f}{S}$ are the
thread and service, respectively, that result from processing the
method of each basic action of the form $f.m$ that the thread denoted
by $p$ performs by the service denoted by $S$.
When the method of a basic action of the form $f.m$ performed by a 
thread is processed by a service, the service changes in accordance with 
the method concerned and affects the thread as follows: the basic action 
turns into the internal action $\Tau$ and the two ways to proceed reduce 
to one on the basis of the reply value produced by the service.

The axioms for the use operators are given in Table~\ref{axioms-use}%
\begin{table}[!t]
\caption{Axioms for the use operators}
\label{axioms-use}
\begin{eqntbl}
\begin{saxcol}
\use{\Stop}{f}{S} = \Stop                            & & \axiom{U1} \\
\use{\DeadEnd}{f}{S} = \DeadEnd                      & & \axiom{U2} \\
\use{(\Tau \bapf x)}{f}{S} =
                          \Tau \bapf (\use{x}{f}{S}) & & \axiom{U3} \\
\use{(\pcc{x}{g.m}{y})}{f}{S} =
\pcc{(\use{x}{f}{S})}{g.m}{(\use{y}{f}{S})}
 & \mif f \neq g                                       & \axiom{U4} \\
\use{(\pcc{x}{f.m}{y})}{f}{S} =
\Tau \bapf (\use{x}{f}{\derive{m}(S)})
                       & \mif \sreply{m}(S) = \True    & \axiom{U5} \\
\use{(\pcc{x}{f.m}{y})}{f}{S} =
\Tau \bapf (\use{y}{f}{\derive{m}(S)})
                       & \mif \sreply{m}(S) = \False   & \axiom{U6} \\
\use{(\pcc{x}{f.m}{y})}{f}{S} = \Tau \bapf \DeadEnd
                       & \mif \sreply{m}(S) = \Blocked & \axiom{U7}
\end{saxcol}
\end{eqntbl}
\end{table}
and the axioms for the apply operators are given in
Table~\ref{axioms-apply}.%
\begin{table}[!t]
\caption{Axioms for the apply operators}
\label{axioms-apply}
\begin{eqntbl}
\begin{saxcol}
\apply{\Stop}{f}{S} = S                              & & \axiom{A1} \\
\apply{\DeadEnd}{f}{S} = \emptyserv                  & & \axiom{A2} \\
\apply{(\Tau \bapf x)}{f}{S} = \apply{x}{f}{S}       & & \axiom{A3} \\
\apply{(\pcc{x}{g.m}{y})}{f}{S} = \emptyserv \hsp{10.2}
                                       & \mif f \neq g & \axiom{A4} \\
\apply{(\pcc{x}{f.m}{y})}{f}{S} = \apply{x}{f}{\derive{m}S}
                       & \mif \sreply{m}(S) = \True    & \axiom{A5} \\
\apply{(\pcc{x}{f.m}{y})}{f}{S} = \apply{y}{f}{\derive{m}S}
                       & \mif \sreply{m}(S) = \False   & \axiom{A6} \\
\apply{(\pcc{x}{f.m}{y})}{f}{S} = \emptyserv
                       & \mif \sreply{m}(S) = \Blocked & \axiom{A7}
\end{saxcol}
\end{eqntbl}
\end{table}
In these tables, $f$ and $g$ stand for arbitrary foci from $\Foci$,
$m$ stands for an arbitrary method from $\Meth$, and $S$ stands for an
arbitrary term of sort $\Serv$.
The axioms simply formalize the informal explanation given above and in
addition stipulate what is the result of use and apply if inappropriate
foci or methods are involved.

The extension of \BTA\ described in this section is a simple version of 
the extension of \BTA\ presented in~\cite{BM09k}.
We have chosen to use the former extension because it is adequate to the
purpose of this paper and it allows a terser survey.

\section{Instruction Sequences Acting on Boolean Registers}
\label{sect-Boolean-register}

In our approach to computational complexity, instruction sequences that 
act on Boolean registers play a key part.
Preceding the presentation of this approach, we describe in this section 
services that make up Boolean registers, introduce special foci that 
serve as names of Boolean registers, and describe the instruction 
sequences that matter to the kinds of complexity classes introduced in 
this paper.

First, we describe services that make up Boolean registers.
The Boolean register services are able to process the following methods:
\begin{itemize}
\item
the \emph{set to true method} $\setbr{\True}$;
\item
the \emph{set to false method} $\setbr{\False}$;
\item
the \emph{get method} $\getbr$.
\end{itemize}
We write $\Methbr$ for the set
$\set{\setbr{\True},\setbr{\False},\getbr}$.
It is assumed that $\Methbr \subseteq \Meth$.

The methods that Boolean register services are able to process can be
explained as follows:
\begin{itemize}
\item
$\setbr{\True}$\,:
the contents of the Boolean register becomes $\True$ and the reply is
$\True$;
\item
$\setbr{\False}$\,:
the contents of the Boolean register becomes $\False$ and the reply is
$\False$;
\item
$\getbr$\,:
nothing changes and the reply is the contents of the Boolean register.
\end{itemize}

For $\Sig{\Services}$, we take the signature that consists of the sorts,
constants and operators that are mentioned in the assumptions with
respect to services made in Section~\ref{sect-TSI} and a constant 
$\BR_b$ of sort $\Serv$ for each $b \in \Bool$.

For $\ServAlg$, we take a minimal $\Sig{\Services}$-algebra that
satisfies the conditions that are mentioned in the assumptions with
respect to services made in Section~\ref{sect-TSI} and the following
conditions for each $b \in \Bool$:
\pagebreak[2]
\begin{ldispl}
\begin{geqns}
\derive{\setbr{\True}}(\BR_{b})  = \BR_{\True}\;,
\\[.5ex]
\derive{\setbr{\False}}(\BR_{b}) = \BR_{\False}\;,
\eqnsep 
\sreply{\setbr{\True}}(\BR_{b})  = \True\;,
\\
\sreply{\setbr{\False}}(\BR_{b}) = \False\;,
\end{geqns}
\qquad
\begin{gceqns}
\derive{\getbr}(\BR_{b}) = \BR_{b}\;,
\\[.5ex]
\derive{m}(\BR_{b}) = \emptyserv
 & \mif m \notin \set{\setbr{\True},\setbr{\False},\getbr}\;,
\eqnsep 
\sreply{\getbr}(\BR_{b}) = b\;,
\\
\sreply{m}(\BR_{b}) = \Blocked
 & \mif m \notin \set{\setbr{\True},\setbr{\False},\getbr}\;.
\end{gceqns}
\end{ldispl}%

In the instruction sequences which concern us in the remainder of this
paper, a number of Boolean registers is used as input registers, a
number of Boolean registers is used as auxiliary registers, and one
Boolean register is used as output register.

It is assumed that $\inbr{1},\inbr{2},\ldots \in \Foci$,
$\auxbr{1},\auxbr{2},\ldots \in \Foci$, and $\outbr \in \Foci$.
These foci play special roles:
\begin{itemize}
\item
for each $i \in \Natpos$, $\inbr{i}$ serves as the name of the Boolean
register that is used as $i$th input register in instruction
sequences;
\item
for each $i \in \Natpos$, $\auxbr{i}$ serves as the name of the Boolean
register that is used as $i$th auxiliary register in instruction
sequences;
\item
$\outbr$ serves as the name of the Boolean register that is used as
output register in instruction sequences.
\end{itemize}
Henceforth, we will write
$\Fociin$ for $\set{\inbr{i} \where i \in \Natpos}$ and
$\Fociaux$ for $\set{\auxbr{i} \where i \in \Natpos}$.

$\ISbr$ is the set of all instruction sequences from $\FIS$ in which all 
plain basic instructions, positive test instructions and negative test 
instructions contain only basic instructions from the set
\begin{ldispl}
\set{f.\getbr    \where f \in \Fociin \union \Fociaux} \union
\set{f.\setbr{b} \where f \in \Fociaux \union \set{\outbr} \Land
                        b \in \Bool}\;;
\end{ldispl}%
$\ISbr$ is the set of all instruction sequences from $\FIS$ that matter
to the kinds of complexity classes which will be introduced in this 
paper.

For each $k,l \in \Nat$, we will write $\RISbr{k}{l}$ for the set of all 
$X \in \ISbr$ that satisfy:
\begin{quote}
\begin{itemize}
\item
primitive instructions of the forms $\auxbr{i}.m$, $\ptst{\auxbr{i}.m}$ 
and $\ntst{\auxbr{i}.m}$ with $i > k$ do not occur in $X$;
\item
primitive instructions of the form $\fjmp{l'}$ with $l' > l$ do not
occur in $X$.
\end{itemize}
\end{quote}
Moreover, for each $k \in \Nat$, we will write $\ARISbr{k}$ for the set
$\Union{l \in \Nat} \RISbr{k}{l}$.
Hence, $\ISbrna$ is the set of all instruction sequences from $\ISbr$ in 
which no auxiliary registers are used, and $\RISbr{0}{0}$ is the set of 
all instruction sequences from $\ISbrna$ in which jump instructions do 
not occur.

\section{The Complexity Classes $\nuc{\IS}{\FN}$}
\label{sect-complclass-PLIS}

In this section, we introduce a kind of non-uniform complexity classes
which includes a counterpart of the complexity class \PTpoly\ in the 
setting of single-pass instruction sequences.

The counterpart of \PTpoly\ defined in this section is denoted by \PLIS.
Because it is isomorphic to the complexity class \PTpoly, we could have
decided to loosely denote this complexity class by \PTpoly\ as well.
The reason why we decided not to denote it by \PTpoly\ finds its origin 
in what we want to achieve with this paper: illustrating an approach to 
non-uniform complexity in which single-pass instruction sequences play a 
key part.
We reserve the use of the name \PTpoly\ to where results obtained in
the setting of Turing machines or the setting of Boolean circuits are
involved.

In the field of computational complexity, it is quite common to study
the complexity of computing functions on finite strings over a binary
alphabet.
Since strings over an alphabet of any fixed size can be efficiently
encoded as strings over a binary alphabet, it is sufficient to consider
only a binary alphabet.
We adopt the set $\Bool$ as preferred binary alphabet.

An important special case of functions on finite strings over a binary
alphabet is the case where the value of functions is restricted to
strings of length $1$.
Such a function is often identified with the set of strings of which it
is the characteristic function.
The set in question is usually called a language or a decision problem.
The identification mentioned above allows of looking at the problem of
computing a function $\funct{f}{\seqof{\Bool}}{\Bool}$ as the problem of
deciding membership of the set
$\set{w \in \seqof{\Bool} \where f(w) = \True}$.

With each function $\funct{f}{\seqof{\Bool}}{\Bool}$, we can associate
an infinite sequence $\indfam{f_n}{n \in \Nat}$ of functions, with
$\funct{f_n}{\Bool^n}{\Bool}$ for every $n \in \Nat$, such that $f_n$ is
the restriction of $f$ to $\Bool^n$  for each $n \in \Nat$.
The complexity of computing such sequences of functions, which we call
Boolean function families, by instruction sequences is our concern in 
the remainder of this paper.
One of the classes of Boolean function families with which we concern us 
is \PLIS, the class of all Boolean function families that can be 
computed by polynomial-length instruction sequences from~$\ISbr$.

An \emph{$n$-ary Boolean function} is a function
$\funct{f}{\Bool^n}{\Bool}$, and
a \emph{Boolean function family} is an infinite sequence
$\indfam{f_n}{n \in \Nat}$ of functions, where $f_n$ is an $n$-ary
Boolean function for each $n \in \Nat$.

A Boolean function family $\indfam{f_n}{n \in \Nat}$ can be identified
with the unique function $\funct{f}{\seqof{\Bool}}{\Bool}$ such that for
each $n \in \Nat$, for each $w \in \Bool^n$, $f(w) = f_n(w)$.
Considering sets of Boolean function families as complexity classes
looks to be most natural when studying non-uniform complexity.
We will make the identification mentioned above only where connections
with classical complexity classes such as \PTpoly\ are made.

\newcommand{\sfuse}{\mathbin{/}}
\newcommand{\sfapply}{\mathbin{\bullet}}
\newcommand{\sfcomp}{\oplus}
\newcommand{\Sfcomp}[2]{\mathop{\mathpalette\nsfcomp{}}_{#1}^{#2}}
\newcommand{\nsfcomp}[1]
  {\raisebox{-.325ex}[1.75ex][.55ex]{\Large $#1 \oplus$}}

Let $n \in \Nat$, let $\funct{f}{\Bool^n}{\Bool}$, and 
let $X \in \ISbr$.
Then $X$ \emph{computes} $f$ if there exists an $l \in \Nat$ such that
for all $b_1,\ldots,b_n \in \Bool$:
\pagebreak[2]
\begin{ldispl}
( \ldots
 (( \ldots
   (\extr{X}
     \useop{\auxbr{1}} \BR_\False) \ldots \useop{\auxbr{l}} \BR_\False)
   \useop{\inbr{1}} \BR_{b_1}) \ldots \useop{\inbr{n}} \BR_{b_n})
 \applyop{\outbr} \BR_\False
\\ \quad {} = \BR_{f(b_1,\ldots,b_n)}\;.\footnotemark
\end{ldispl}%
\footnotetext
{In the extension of \BTA\ presented in~\cite{BM09k}, which has a sort 
 of (named) service families and a service family composition operator
 ($\sfcomp$), the left-hand side of this equation can be written as 
 follows:
 $(\extr{X} \sfuse
  ((\Sfcomp{i=1}{n} \inbr{i}.\BR_{b_i}) \sfcomp
   (\Sfcomp{j=1}{l} \auxbr{j}.\BR_{\False}))) \sfapply
  \outbr.\BR_{\False}$.}%
Moreover, let $\IS \subseteq \ISbr$ and 
$\FN \subseteq \set{h \where \funct{h}{\Nat}{\Nat}}$.
Then $\nuc{\IS}{\FN}$ is the class of all Boolean function families
$\indfam{f_n}{n \in \Nat}$ that satisfy:
\begin{quote}
there exists an $h \in \FN$ such that
for all $n \in \Nat$ there exists an $X \in \IS$ such that
$X$ computes $f_n$ and $\psize(X) \leq h(n)$.
\end{quote}

Henceforth, we will write $\poly$ for the set 
$\set{h \where 
      \funct{h}{\Nat}{\Nat} \Land h \mathrm{\,is \,polynomial}}$.
We are primarily interested in the complexity class $\PLIS$,%
\footnote
{In precursors of this paper, the temporary name 
 $\mathrm{P}^*$ is used for the complexity class $\PLIS$ 
 (see e.g.~\cite{BM08g}).}
but we will also pay attention to other instantiations of the general 
definition just given.

The question arises whether all $n$-ary Boolean functions can be
computed by an instruction sequence from $\ISbr$. 
This question can answered in the affirmative.
They can even be computed, without using auxiliary Boolean registers,
by an instruction sequence that contains no other
jump instructions than $\fjmp{2}$.
\begin{theorem}
\label{theorem-comput-boolfunc}
For each $n \in \Nat$, for each $n$-ary Boolean function
$\funct{f}{\Bool^n}{\Bool}$, there exists an $X \in \ISbrna$ in which
no other jump instruction than $\fjmp{2}$ occurs such that $X$
computes $f$ and $\psize(X) = O(2^n)$.%
\footnote
{Theorem~\ref{theorem-comput-boolfunc} sharpens the result found in 
 precursors of this paper (see e.g.~\cite{BM08g}).
 We owe the sharpened result to Inge Bethke from the University of 
 Amsterdam.}
\end{theorem}
\begin{proof}
Let $\inseq_n$ be the function from the set of all $n$-ary Boolean
function $\funct{f}{\Bool^n}{\Bool}$ to $\ISbrna$ defined by induction
on $n$ as follows:
\begin{ldispl}
\begin{aeqns}
\inseq_0(f) & = &
\left\{
\begin{array}[c]{@{}l@{\quad}l@{}}
\ntst{\outbr.\setbr{\True}}  \conc \fjmp{2} \conc \halt
 & \mif f() = \True \\
\ptst{\outbr.\setbr{\False}} \conc \fjmp{2} \conc \halt
 & \mif f() = \False\;,
\end{array}
\right.
\eqnsep
\inseq_{n+1}(f) & = &
\ntst{\inbr{n{+}1}.\getbr} \conc \fjmp{2} \conc
\inseq_n(f_\True) \conc \inseq_n(f_\False)\;,
\end{aeqns}
\end{ldispl}%
where for each $\funct{f}{\Bool^{n+1}}{\Bool}$ and $b \in \Bool$,
$\funct{f_b}{\Bool^n}{\Bool}$ is defined as follows:
\begin{ldispl}
f_b(b_1,\ldots,b_n) = f(b_1,\ldots,b_n,b)\;.
\end{ldispl}%
It is easy to prove by induction on $n$ that
$\extr{\fjmp{2} \conc \inseq_n(f_\True) \conc X} = \extr{X}$.
Using this fact, it is easy to prove by induction on $n$ that
$\inseq_n(f)$ computes $f$.
Moreover, it is easy to see that $\psize(\inseq_n(f)) = O(2^n)$.
\qed
\end{proof}

Henceforth, we will use the notation $\nuc{\IS}{O(f(n))}$ for the 
complexity class
$\nuc{\IS}{\set{h \where \funct{h}{\Nat}{\Nat} \Land h(n) = O(f(n))}}$.
This notation is among other things used in the following corollary of
Theorem~\ref{theorem-comput-boolfunc}.
\begin{corollary}
\label{corollary-comput-boolfunc-1}
All Boolean function families belong to $\nuc{\ISbrna}{O(2^n)}$.
\end{corollary}

In the proof of Theorem~\ref{theorem-comput-boolfunc}, the instruction
sequences yielded by the function $\inseq_n$ contain the jump
instruction $\fjmp{2}$.
Each occurrence of $\fjmp{2}$ belongs to a jump chain ending in
the instruction sequence
$\ntst{\outbr.\setbr{\True}} \conc \fjmp{2} \conc \halt$ or
the instruction sequence
$\ptst{\outbr.\setbr{\False}} \conc \fjmp{2} \conc \halt$.
Therefore, each occurrence of $\fjmp{2}$ can safely be replaced by the
instruction $\ptst{\outbr.\setbr{\False}}$, which like $\fjmp{2}$ skips
the next instruction.
This leads to the following corollary.
\begin{corollary}
\label{corollary-comput-boolfunc-2}
$\nuc{\RISbr{0}{0}}{O(2^n)} = \nuc{\ISbrna}{O(2^n)} = \nuc{\ISbr}{O(2^n)}$.
\end{corollary}

We consider the proof of Theorem~\ref{theorem-comput-boolfunc} once
again.
Because the content of the Boolean register concerned is initially
$\False$, the question arises whether $\outbr.\setbr{\False}$ can be
dispensed with in instruction sequences computing Boolean functions.
This question can be answered in the affirmative if we permit the use of
auxiliary Boolean registers.
\begin{theorem}
\label{theorem-comput-output}
Let $n \in \Nat$, let $\funct{f}{\Bool^n}{\Bool}$, and
let $X \in \ISbr$ be such that $X$ computes $f$.
Then there exists an $Y \in \ISbr$ in which the basic instruction
$\outbr.\setbr{\False}$ does not occur such that $Y$ computes $f$ and
$\psize(Y)$ is linear in $\psize(X)$.
\end{theorem}
\begin{proof}
Let $o \in \Natpos$ be such that the basic instructions
$\auxbr{o}.\setbr{\True}$, $\auxbr{o}.\setbr{\False}$, and
$\auxbr{o}.\getbr$ do not occur in $X$.
Let $X'$ be obtained from $X$ by replacing each occurrence of the focus
$\outbr$ by $\auxbr{o}$.
Suppose that $X' = u_1 \conc \ldots \conc u_k$.
Let $Y$ be obtained from $u_1 \conc \ldots \conc u_k$ as follows:
\begin{enumerate}
\item
stop if $u_1 \equiv \halt$;
\item
stop if there exists no $j \in [2,k]$ such that
$u_{j-1} \not\equiv \outbr.\setbr{\True}$ and $u_j \equiv \halt$;
\item
find the least $j \in [2,k]$ such that
$u_{j-1} \not\equiv \outbr.\setbr{\True}$ and $u_j \equiv \halt$;
\item
replace $u_j$ by
$\ptst{\auxbr{o}.\getbr} \conc \outbr.\setbr{\True} \conc \halt$;
\item
for each $i \in [1,k]$, replace $u_i$ by $\fjmp{l{+}2}$ if
$u_i \equiv \fjmp{l}$ and $i < j < i + l$;
\item
repeat the preceding steps for the resulting instruction sequence.
\end{enumerate}
It is easy to prove by induction on $k$ that the Boolean function
computed by $X$ and the Boolean function computed by $Y$ are the
same.
Moreover, it is easy to see that $\psize(Y) < 3 \mul \psize(X)$.
Hence, $\psize(Y)$ is linear in $\psize(X)$.
\qed
\end{proof}

The following proposition gives an upper bound for the number of 
instruction sequences from $\RISbr{k-1}{k-1}$ of length $k$ that compute 
an $n$-ary Boolean function.
From each instruction sequence from $\ISbr$ of length $k$ that computes 
an $n$-ary Boolean function, we can obtain an instruction sequence from 
$\RISbr{k-1}{k-1}$ of length $k$ that computes the same $n$-ary Boolean 
function by replacement of the primitive instructions that are not 
permitted in $\RISbr{k-1}{k-1}$.
Moreover, each $n$-ary Boolean function that can be computed by an 
instruction sequence from $\ISbr$ of length less than $k$, can also be 
computed by an instruction sequence from $\ISbr$ of length~$k$.
\begin{proposition}
\label{prop-inseq-number}
For each $k \in \Natpos$ and $n \in \Nat$, the number of instruction 
sequences from $\RISbr{k-1}{k-1}$ of length $k$ that compute an $n$-ary 
Boolean function is not greater than $(3n + 10k - 2)^k$.
\end{proposition}
\begin{proof}
The set of basic instructions from which the plain basic instructions, 
positive test instructions and negative test instructions occurring in 
the instruction sequences concerned are built consists of
$n$ basic instructions of the form $\inbr{i}.\getbr$, 
$k - 1$ basic instructions of each of the forms 
$\auxbr{i}.\setbr{\True}$, $\auxbr{i}.\setbr{\False}$ and 
$\auxbr{i}.\getbr$, and
the two basic instructions $\outbr.\setbr{\True}$ and 
$\outbr.\setbr{\False}$.
Moreover, there are $k$ different jump instructions that may occur and 
one termination instruction.
This means that there are $3(n + 3(k - 1) + 2) + k + 1 = 3n + 10k - 2$ 
different primitive instructions that may occur in these instruction 
sequences.
Hence, the number of instruction sequences concerned is not greater than 
$(3n + 10k - 2)^k$.
\qed
\end{proof}

Theorem~\ref{theorem-comput-boolfunc} states that all $n$-ary Boolean 
functions can be computed by an instruction sequence from $\ISbr$ whose 
length is exponential in $n$.
The following theorem shows that, for large enough $n$, not all $n$-ary 
Boolean functions can be computed by an instruction sequence from 
$\ISbr$ whose length is polynomial in $n$.
\begin{theorem}
\label{theorem-comput-large-inseq}
For each $n \in \Nat$ with $n > 11$, there exists a $n$-ary Boolean 
function $\funct{f}{\Bool^n}{\Bool}$ such that, for each $X \in \ISbr$ 
that computes $f$, $\psize(X) > \lfloor 2^n / n \rfloor$.
\end{theorem}
\begin{proof}
Let $n \in \Nat$ be such that $n > 11$.
By Proposition~\ref{prop-inseq-number} and the remarks im\-mediately 
preceding Proposition~\ref{prop-inseq-number}, the number of $n$-ary 
Boolean functions that can be computed by instruction 
sequences from $\ISbr$ of length less than or equal to $k$ is not 
greater than $(3n + 10k - 2)^k$.
For $k = \lfloor 2^n / n \rfloor$, this number is not greater than 
$(3n + 10\lfloor 2^n / n \rfloor - 2)^{\lfloor 2^n / n \rfloor}$.
We have that
$(3n + 10\lfloor 2^n / n \rfloor - 2)^{\lfloor 2^n / n \rfloor} \leq 
 (3n + 10(2^n / n) - 2)^{2^n / n} < (11(2^n / n))^{2^n / n} = 
 (11/n)^{2^n / n} \cdot (2^n)^{2^n / n} = 
 (11/n)^{2^n / n} \cdot 2^{(2^n)} < 2^{(2^n)}$.
Here, we have used the given that $n > 11$ in the second step and the 
last step.
So there exist less than $2^{(2^n)}$ $n$-ary Boolean functions that can 
be computed by instruction sequences from $\ISbr$ of length less than or 
equal to $\lfloor 2^n / n \rfloor$, whereas there exist $2^{(2^n)}$ 
$n$-ary Boolean functions.
Hence, there exists an $n$-ary Boolean function that cannot be computed 
by an instruction sequence from $\ISbr$ of length less than or equal to
$\lfloor 2^n / n \rfloor$.
\qed
\end{proof}

Theorem~\ref{theorem-comput-large-inseq} gives rise to the following
corollary concerning \PLIS.
\begin{corollary}
\label{corollary-comput-large-inseq}
$\PLIS \subset \nuc{\ISbr}{O(2^n)}$.
\end{corollary}

Theorem~\ref{theorem-comput-large-inseq} will be used in the proof of 
the following hierarchy theorem for \PLIS.
\begin{theorem}
\label{theorem-hierarchy-PLIS}
For each $k \in \Nat$, 
$\nuc{\ISbr}{O(n^k)} \subset \nuc{\ISbr}{O(n^{k+1})}$.
\end{theorem}
\begin{proof}
Let $\indfam{f_n}{n \in \Nat}$ be a Boolean function family such that,
for each $n > 11$, for each $X \in \ISbr$ that computes $f_n$, 
$\psize(X) > \lfloor 2^n / n \rfloor$.
Such a Boolean function family exists by 
Theorem~\ref{theorem-comput-large-inseq}.
Let $k \in \Nat$, and let $\indfam{g_n}{n \in \Nat}$ be the Boolean 
function family such that, for each $n \in \Nat$, 
$g_n(b_1,\ldots,b_n) = f_n(b_1,\ldots,b_n)$ if $n < 2^{k+3}$ and
$g_n(b_1,\ldots,b_n) = 
 f_{\lceil \log((k+2)n^{k+1}) \rceil}
 (b_1,\ldots,b_{\lceil \log((k+2)n^{k+1}) \rceil})$ if $n \geq 2^{k+3}$.
Then $\indfam{g_n}{n \in \Nat} \in \nuc{\ISbr}{O(n^{k+1})}$ by 
Theorem~\ref{theorem-comput-boolfunc}.
\sloppy
Moreover, for each $n \geq 2^{k+3}$, for each $Y \in \ISbr$ that 
computes $g_n$, 
$\psize(Y) >
 \lfloor 2^{\lceil \log((k+2)n^{k+1}) \rceil} / 
 \lceil \log((k+2)n^{k+1}) \rceil \rfloor \geq 
 \lfloor 2^{\log((k+2)n^{k+1})} / (\log((k+2)n^{k+1}) + 1) \rfloor \geq 
 \lfloor (k+2)n^{k+1} / (k+2)\log(n) \rfloor = 
 \lfloor n^{k+1} / \log(n) \rfloor$.
Here, we have used the given that $n \geq 2^{k+3}$ in the last step but 
one.
From the fact that, for all $m \in \Nat$, there exists an $n \in \Nat$ 
such that $n > m \log(n)$, it follows that not 
$\lfloor n^{k+1} / \log(n) \rfloor = O(n^k)$.
Hence, $\indfam{g_n}{n \in \Nat} \notin \nuc{\ISbr}{O(n^k)}$.
\qed
\end{proof}

As a corollary of the fact that
$\poly =
 \Union{k \in \Nat}
  \set{h \where \funct{h}{\Nat}{\Nat} \Land h(n) = O(n^k)}$,
the general definition of the non-uniform complexity classes 
$\nuc{\IS}{\FN}$, and Theorem~\ref{theorem-hierarchy-PLIS}, we have the
following result.
\begin{corollary}
\label{corollary-hierarchy-PLIS}
For each $k \in \Nat$, $\PLIS \not\subseteq \nuc{\ISbr}{O(n^k)}$.
\end{corollary}

\section{Instruction Sequences, Boolean Formulas and Circuits}
\label{sect-boolform-and-boolcirc}

In this section, we investigate connections of single-pass instruction 
sequences with Boolean formulas and Boolean circuits which are relevant 
to non-uniform complexity and show that \PLIS\ coincides with \PTpoly.
The definitions of Boolean circuits, \PTpoly\ and related notions on 
which some results in this section and the coming ones are based are the 
definitions from Chapter~6 of~\cite{AB09a}.

First, we dwell on obtaining instruction sequences that compute the
Boolean functions induced by Boolean formulas from the Boolean formulas
concerned.

Hereafter, we will write $\phi(b_1,\ldots,b_n)$, where $\phi$ is a 
Boolean formula containing the variables $v_1,\ldots,v_n$ and
$b_1,\ldots,b_n \in \Bool$, to indicate that $\phi$ is satisfied by
the assignment $\sigma$ to the variables $v_1,\ldots,v_n$ defined by
$\sigma(v_1) = b_1$, \ldots, $\sigma(v_n) = b_n$.

Let $\phi$ be a Boolean formula containing the variables
$v_1,\ldots,v_n$.
Then the Boolean function \emph{induced} by $\phi$ is the $n$-ary 
Boolean function $f$ defined by $f(b_1,\ldots,b_n) = \True$ iff 
$\phi(b_1,\ldots,b_n)$.

The Boolean function induced by a \CNF-formula can be computed, without
using auxiliary Boolean registers, by an instruction sequence that 
contains no other jump instructions than $\fjmp{2}$ and whose length is 
linear in the size of the \CNF-formula.
\begin{proposition}
\label{prop-comput-cnfform}
For each \CNF-formula $\phi$, there exists an $X \in \ISbrna$ in
which no other jump instruction than $\fjmp{2}$ occurs such that $X$
computes the Boolean function induced by $\phi$ and $\psize(X)$ is
linear in the size of $\phi$.
\end{proposition}
\begin{proof}
Let $\inseqcnf$ be the function from the set of all \CNF-formulas
containing the variables $v_1,\ldots,v_n$ to $\ISbrna$ defined as 
follows:
\begin{ldispl}
\inseqcnf
 \bigl(\LAND{i \in [1,m]} \LOR{j \in [1,n_i]} \xi_{ij}\bigr) =
 {} \\ \quad
\inseqcnf'(\xi_{11}) \conc \ldots \conc \inseqcnf'(\xi_{1n_1}) \conc
\ptst{\outbr.\setbr{\False}} \conc \fjmp{2} \conc \halt \conc
\\ \qquad \vdots \\ \quad
\inseqcnf'(\xi_{m1}) \conc \ldots \conc \inseqcnf'(\xi_{mn_m}) \conc
\ptst{\outbr.\setbr{\False}} \conc \fjmp{2} \conc \halt \conc
\ptst{\outbr.\setbr{\True}} \conc \halt\;,
\end{ldispl}%
where
\begin{ldispl}
\begin{aeqns}
\inseqcnf'(v_k)     & = & \ptst{\inbr{k}.\getbr} \conc \fjmp{2}\;,
\\
\inseqcnf'(\Lnot v_k)& = & \ntst{\inbr{k}.\getbr} \conc \fjmp{2}\;.
\end{aeqns}
\end{ldispl}%
It is easy to see that no other jump instruction than $\fjmp{2}$ occurs
in $\inseqcnf(\phi)$.
Recall that a disjunction is satisfied if one of its disjuncts is
satisfied and a conjunction is satisfied if each of its conjuncts is
satisfied.
Using these facts, it is easy to prove by induction on the number of
clauses in a \CNF-formula, and in the basis step by induction on the
number of literals in a clause, that $\inseqcnf(\phi)$ computes the
Boolean function induced by $\phi$.
Moreover, it is easy to see that $\psize(\inseqcnf(\phi))$ is linear
in the size of $\phi$.
\qed
\end{proof}

In the proof of Proposition~\ref{prop-comput-cnfform}, it is shown that
the Boolean function induced by a \CNF-formula can be computed, without
using auxiliary Boolean registers, by an instruction sequence that 
contains no other jump instructions than $\fjmp{2}$.
However, the instruction sequence concerned contains the termination 
instruction more than once and both $\outbr.\setbr{\True}$ and
$\outbr.\setbr{\False}$.
This raises the question whether further restrictions are possible.
We have a negative result.
\begin{proposition}
\label{prop-comput-negative}
Let $\phi$ be the Boolean formula $v_1 \Land v_2 \Land v_3$.
Then there does not exist an $X \in \RISbr{0}{0}$ in which the 
termination instruction does not occur more than once and the basic 
instruction $\outbr.\setbr{\False}$ does not occur such that $X$ 
computes the Boolean function induced by $\phi$.
\end{proposition}
\begin{proof}
Suppose that $X = u_1 \conc \ldots \conc u_k$ is an instruction
sequence from $\RISbr{0}{0}$ satisfying the restrictions and computing 
the Boolean function induced by $\phi$.
Consider the smallest $l \in [1,k]$ such that $u_l$ is either
$\outbr.\setbr{\True}$, $\ptst{\outbr.\setbr{\True}}$ or
$\ntst{\outbr.\setbr{\True}}$ (there must be such an $l$).
Because $\phi$ is not satisfied by all assignments to the variables
$v_1,v_2,v_3$, it cannot be the case that $l = 1$.
In the case where $l > 1$, for each $i \in [1,l-1]$, $u_i$ is either
$\inbr{j}.\getbr$, $\ptst{\inbr{j}.\getbr}$ or $\ntst{\inbr{j}.\getbr}$
for some $j \in \set{1,2,3}$.
This implies that, for each $i \in [0,l-1]$, there exists a basic
Boolean formula $\psi_i$ over the variables $v_1,v_2,v_3$ that is unique
up to logical equivalence such that, for each $b_1,b_2,b_3 \in \Bool$,
if the initial states of the Boolean registers named $\inbr{1}$,
$\inbr{2}$ and $\inbr{3}$ are $b_1$, $b_2$ and $b_3$, respectively, then
$u_{i+1}$ will be executed iff $\psi_i(b_1,b_2,b_3)$.
We have that $\psi_0 \Liff \True$ and,
for each $i \in [1,l-1]$,
$\psi_i \Liff (\psi_{i-1} \Limpl \True)$ if
 $u_i \equiv \inbr{j}.\getbr$,
$\psi_i \Liff (\psi_{i-1} \Limpl v_j)$ if
 $u_i \equiv \ptst{\inbr{j}.\getbr}$, and
$\psi_i \Liff (\psi_{i-1} \Limpl \Lnot v_j)$ if
 $u_i \equiv \ntst{\inbr{j}.\getbr}$.
Hence, for each $i \in [0,l-1]$, $\psi_i \Limpl \phi$ implies
$\True \Limpl \phi$ or $v_j \Limpl \phi$ or $\Lnot v_j \Limpl \phi$
for some $j \in \set{1,2,3}$.
Because the latter three Boolean formulas are no tautologies,
$\psi_i \Limpl \phi$ is no tautology either.
This means that, for each $i \in [1,l-1]$,
$\psi_i \Limpl \phi$ is not satisfied by all assignments to the
variables $v_1,v_2,v_3$.
Hence, $X$ cannot exist.
\qed
\end{proof}

According to Proposition~\ref{prop-comput-cnfform}, the Boolean function
induced by a \CNF-formula can be computed, without using auxiliary
Boolean registers, by an instruction sequence that contains no other 
jump instructions than $\fjmp{2}$ and whose length is linear in the size 
of the formula.
If we permit arbitrary jump instructions, this result generalizes from
\CNF-formulas to arbitrary basic Boolean formulas, i.e.\ Boolean
formulas in which no other connectives than $\Lnot$, $\Lor$ and $\Land$
occur.
\begin{proposition}
\label{prop-comput-boolform}
For each basic Boolean formula $\phi$, there exists an $X \in \ISbrna$
in which the basic instruction $\outbr.\setbr{\False}$ does not occur
such that $X$ computes the Boolean function induced by $\phi$ and
$\psize(X)$ is linear in the size of $\phi$.
\end{proposition}
\begin{proof}
Let $\inseqf$ be the function from the set of all basic Boolean formulas
containing the variables $v_1,\ldots,v_n$ to $\ISbrna$ defined as 
follows:
\begin{ldispl}
\inseqf(\phi) =
\inseqf'(\phi) \conc \ptst{\outbr.\setbr{\True}} \conc \halt\;,
\end{ldispl}%
where
\begin{ldispl}
\begin{aeqns}
\inseqf'(v_k) & = & \ptst{\inbr{k}.\getbr}\;,
\\
\inseqf'(\Lnot \phi) & = & \inseqf'(\phi) \conc \fjmp{2}\;,
\\
\inseqf'(\phi \Lor \psi) & = &
\inseqf'(\phi) \conc
\fjmp{\psize(\inseqf'(\psi)){+}1} \conc \inseqf'(\psi)\;,
\\
\inseqf'(\phi \Land \psi) & = &
\inseqf'(\phi) \conc \fjmp{2} \conc
\fjmp{\psize(\inseqf'(\psi)){+}2} \conc \inseqf'(\psi)\;.
\end{aeqns}
\end{ldispl}%
Using the same facts about disjunctions and conjunctions as in the proof
of Proposition~\ref{prop-comput-cnfform}, it is easy to prove by
induction on the structure of $\phi$ that $\inseqf(\phi)$ computes the
Boolean function induced by $\phi$. \sloppy
Moreover, it is easy to see that $\psize(\inseqf(\phi))$ is linear in the
size of $\phi$.
\qed
\end{proof}

In the next proposition, we consider Boolean circuits instead of Boolean 
formulas.

Let $C$ be a Boolean circuit with $n$ input nodes and a single output 
node.
Then the Boolean function \emph{induced} by $C$ is the $n$-ary Boolean 
function $f$ defined by $f(b_1,\ldots,b_n) = C(b_1,\ldots,b_n)$, where
$C(b_1,\ldots,b_n)$ denotes the output of $C$ on input 
$(b_1,\ldots,b_n)$.

Because Boolean formulas can be looked upon as Boolean circuits with a 
single output node in which all gates have out-degree~$1$, the question 
arises whether Proposition~\ref{prop-comput-boolform} generalizes from 
Boolean formulas to Boolean circuits with a single output node.
This question can be answered in the affirmative if we permit the use of
auxiliary Boolean registers.
\begin{proposition}
\label{prop-comput-boolcirc}
For each Boolean circuit $C$ with a single output node that contains no 
other gates than $\Lnot$-gates, $\Lor$-gates and $\Land$-gates, there 
exists an $X \in \ISbr$ in which the basic instruction 
$\outbr.\setbr{\False}$ does not occur such that $X$ computes the 
Boolean function induced by $C$ and $\psize(X)$ is linear in the size of 
$C$.
\end{proposition}
\begin{proof}
Let $\inseqc$ be the function from the set of all Boolean circuits with
input nodes $\inode{1},\ldots,\inode{n}$, gates 
$\gate{1},\ldots,\gate{m}$ and a single output node $\onode$ to 
$\ISbrna$ defined as follows:
\begin{ldispl}
\inseqc(C) =
\inseqc'(\gate{1}) \conc \ldots \conc \inseqc'(\gate{m}) \conc
\ptst{\auxbr{m}.\getbr} \conc
\ptst{\outbr.\setbr{\True}} \conc \halt\;,
\end{ldispl}%
where
\begin{ldispl}
\inseqc'(\gate{k}) = {} \\ \quad
\inseqc''(\pnode) \conc \fjmp{2} \conc \ptst{\auxbr{k}.\setbr{\True}}
\\ \quad
\mbox
 {if $\gate{k}$ is a $\Lnot$-gate with direct preceding node $\pnode$}\;,
\\
\inseqc'(\gate{k}) = {} \\ \quad
\inseqc''(\pnode) \conc \fjmp{2} \conc
\inseqc''(\pnode') \conc \ptst{\auxbr{k}.\setbr{\True}}
\\ \quad
\mbox
 {if $\gate{k}$ is a $\Lor$-gate with direct preceding nodes $\pnode$
  and $\pnode'$}\;,
\\
\inseqc'(\gate{k}) = {} \\ \quad
\inseqc''(\pnode) \conc \fjmp{2} \conc \fjmp{3} \conc
\inseqc''(\pnode') \conc \ptst{\auxbr{k}.\setbr{\True}}
\\ \quad
\mbox
 {if $\gate{k}$ is a $\Land$-gate with direct preceding nodes $\pnode$
  and $\pnode'$}\;,
\end{ldispl}%
and
\begin{ldispl}
\begin{aceqns}
\inseqc''(\inode{k}) & = & \ptst{\inbr{k}.\getbr}\;,
\\
\inseqc''(\gate{k}) & = & \ptst{\auxbr{k}.\getbr}\;.
\end{aceqns}
\end{ldispl}%
Using the same facts about disjunctions and conjunctions as in the
proofs of Propositions~\ref{prop-comput-cnfform}
and~\ref{prop-comput-boolform}, it is easy to prove by induction on
the depth of $C$ that $\inseqc(C)$ computes the Boolean function induced
by $C$ if $\gate{1},\ldots,\gate{m}$ is a topological sorting of the
gates of $C$. \sloppy
Moreover, it is easy to see that $\psize(\inseqc(C))$ is linear in the
size of $C$.
\qed
\end{proof}

\PLIS\ includes Boolean function families that correspond to
uncomput\-able functions from $\seqof{\Bool}$ to $\Bool$.
Take an undecidable set $N \subseteq \Nat$ and consider the Boolean
function family $\indfam{f_n}{n \in \Nat}$ with, for each $n \in \Nat$,
$\funct{f_n}{\Bool^n}{\Bool}$ defined by
\begin{ldispl}
\begin{gceqns}
f_n(b_1,\ldots,b_n) = \True  & \mif n \in N\;,
\\
f_n(b_1,\ldots,b_n) = \False & \mif n \notin N\;.
\end{gceqns}
\end{ldispl}%
For each $n \in N$, $f_n$ is computed by the instruction sequence
$\outbr.\setbr{\True} \conc \halt$.
For each $n \notin N$, $f_n$ is computed by the instruction sequence
$\outbr.\setbr{\False} \conc \halt$.
The length of these instruction sequences is constant in $n$.
Hence, $\indfam{f_n}{n \in \Nat}$ is in \PLIS.
However, the corresponding function $\funct{f}{\seqof{\Bool}}{\Bool}$ is
clearly uncomputable.
This reminds of the fact that \PTpoly\ includes uncomputable functions
from $\seqof{\Bool}$ to $\Bool$.

It happens that \PLIS\ and \PTpoly\ coincide, provided that we identify
each Boolean function family $\indfam{f_n}{n \in \Nat}$ with the unique
function $\funct{f}{\seqof{\Bool}}{\Bool}$ such that for each
$n \in \Nat$, for each $w \in \Bool^n$, $f(w) = f_n(w)$.
\begin{theorem}
\label{theorem-PLIS-is-PTpoly}
$\PLIS = \PTpoly$.
\end{theorem}
\begin{proof}
We will prove the inclusion $\PTpoly \subseteq \PLIS$ using the
definition of $\PTpoly$ in terms of Boolean circuits and we will prove
the inclusion $\PLIS \subseteq \PTpoly$ using the characterization of
$\PTpoly$ in terms of Turing machines that take advice 
(see e.g.\ Chapter~6 of~\cite{AB09a}).

$\PTpoly \subseteq \PLIS$:
Suppose that $\indfam{f_n}{n \in \Nat}$ in $\PTpoly$.
Then, for all $n \in \Nat$, there exists a Boolean circuit $C$ such that
$C$ computes $f_n$ and the size of $C$ is polynomial in $n$.
For each $n \in \Nat$, let $C_n$ be such a $C$.
From Proposition~\ref{prop-comput-boolcirc} and the fact that linear in
the size of $C_n$ implies polynomial in $n$, it follows that each
Boolean function family in $\PTpoly$ is also in $\PLIS$.

$\PLIS \subseteq \PTpoly$:
Suppose that $\indfam{f_n}{n \in \Nat}$ in $\PLIS$.
Then, for all $n \in \Nat$, there exists an $X \in \ISbr$ such that $X$
computes $f_n$ and $\psize(X)$ is polynomial in $n$.
For each $n \in \Nat$, let $X_n$ be such an $X$.
Then $f$ can be computed by a Turing machine that, on an input of size
$n$, takes a binary description of $X_n$ as advice and then just
simulates the execution of $X_n$.
It is easy to see that under the assumption that 
$X_n \in \RISbr{k-1}{k-1}$, where $k = \psize(X_n)$, the size of the 
description of $X_n$ and the number of steps that it takes to simulate 
the execution of $X_n$ are both polynomial in $n$.
We can make this assumption without loss of generality (see the remarks
immediately preceding Proposition~\ref{prop-inseq-number}).
Hence, each Boolean function family in $\PLIS$ is also in $\PTpoly$.
\qed
\end{proof}

It is unknown to us whether $\nuc{\ARISbr{k}}{\poly}$ is different from 
\PLIS\ for all $k \in \Nat$ (see also Section~\ref{sect-concl}).

\section{Non-uniform Super-polynomial Complexity Conjecture}
\label{sect-conjecture}

In this section, we formulate a complexity conjecture which is a
counterpart of the well-known complexity theoretic conjecture that
$\NPT \not\subseteq \PTpoly$ in the current setting.
The definitions of \NPT, \NPT-hardness, \NPT-completeness and related 
notions on which some results in this section and the coming ones are 
based are the ones from Chapter~2 of~\cite{AB09a}.

The counterpart of the conjecture that $\NPT \not\subseteq \PTpoly$ 
formulated in this section corresponds to the conjecture that
$\iiiSAT \notin \PTpoly$.
By the \NPT-completeness of \iiiSAT, $\iiiSAT \notin \PTpoly$ is 
equivalent to $\NPT \not\subseteq \PTpoly$.
If the conjecture that $\NPT \not\subseteq \PTpoly$ is right, then the
conjecture that $\NPT \neq \PT$ is right as well.

To formulate the conjecture, we need a Boolean function family
$\indfam{\iiiSATC_n}{n \in \Nat}$ that corresponds to \iiiSAT.
We obtain this Boolean function family by encoding \iiiCNF-formulas as
sequences of Boolean values.

We write $\ndisj(k)$ for
$\binom{2k}{1} + \binom{2k}{2} + \binom{2k}{3}$.%
\footnote
{As usual, we write $\binom{k}{l}$ for the number of $l$-element
 subsets of a $k$-element set.
}
$\ndisj(k)$ is the number of combinations of at most $3$
elements from a set with $2k$ elements.
Notice that $\ndisj(k) = (4k^3 + 5k) / 3$.

It is assumed that a countably infinite set $\set{v_1,v_2,\ldots}$ of
propositional variables has been given.
Moreover, it is assumed that a family of bijections
\begin{ldispl}
\indfam
 {\funct{\alpha_k}{[1,\ndisj(k)]}
        {\set{L \subseteq \set{v_1,\Lnot v_1,\ldots,v_k,\Lnot v_k} \where
              1 \leq \card(L) \leq 3}}}
 {k \in \Nat}
\end{ldispl}%
has been given that satisfies the following two conditions:
\begin{ldispl}
\Forall{i \in \Nat}
 {\Forall{j \in [1,\ndisj(i)]}{{\alpha_i}^{-1}(\alpha_{i+1}(j)) = j}}\;,
\\
\alpha\; \mbox{is polynomial-time computable}\;,
\end{ldispl}%
\par
where
$\funct{\alpha}{\Natpos}
       {\set{L \subseteq \set{v_1,\Lnot v_1,v_2,\Lnot v_2,\ldots} \where
             1 \leq \card(L) \leq 3}}$
is defined by
\begin{ldispl}
\mbox{} \hspace*{\leftmargin}
\alpha(i) = \alpha_{\min \set{j \where i \in [1,\ndisj(j)]}}(i)\;.
\end{ldispl}%
The function $\alpha$ is well-defined owing to the first condition
on $\indfam{\alpha_k}{k \in \Nat}$.
The second condition is satisfiable, but it is not satisfied by all
$\indfam{\alpha_k}{k \in \Nat}$ satisfying the first condition.

The basic idea underlying the encoding of \iiiCNF-formulas as sequences
of Boolean values is as follows:
\begin{itemize}
\item
if $n = \ndisj(k)$ for some $k \in \Nat$, then the input of $\iiiSATC_n$
consists of one Boolean value for each disjunction of at most three
literals from the set $\set{v_1,\Lnot v_1,\ldots,v_k,\Lnot v_k}$;
\item
each Boolean value indicates whether the corresponding disjunction
occurs in the encoded \iiiCNF-formula;
\item
if $\ndisj(k) < n < \ndisj(k+1)$ for some $k \in \Nat$, then only the
first $\ndisj(k)$ Boolean values form part of the encoding.
\end{itemize}

For each $n \in \Nat$, $\funct{\iiiSATC_n}{\Bool^n}{\Bool}$ is
defined as follows:
\begin{itemize}
\item
if $n = \ndisj(k)$ for some $k \in \Nat$:
\begin{ldispl}
\iiiSATC_n(b_1,\ldots,b_n) = \True
\;\;\;\mathrm{iff}\;\;\;
\displaystyle
\LAND{i \in [1,n] \;\mathrm{s.t.}\; b_i = \True}
 \LOR{} \alpha_k(i)
\;\;\mathrm{is\;satisfiable}\;,
\end{ldispl}%
where $k$ is such that $n = \ndisj(k)$;
\item
if $\ndisj(k) < n < \ndisj(k+1)$ for some $k \in \Nat$:
\begin{ldispl}
\iiiSATC_n(b_1,\ldots,b_n) =
\iiiSATC_{\ndisj(k)}(b_1,\ldots,b_{\ndisj(k)})\;,
\end{ldispl}%
where $k$ is such that $\ndisj(k) < n < \ndisj(k+1)$.
\end{itemize}

Because $\indfam{\alpha_k}{k \in \Nat}$ satisfies the condition that
${\alpha_i}^{-1}(\alpha_{i+1}(j)) = j$ for all $i \in \Nat$ and
$j \in [1,\ndisj(i)]$, we have for each $n \in \Nat$, for all
$b_1,\ldots,b_n \in \Bool$:
\begin{ldispl}
\iiiSATC_n(b_1,\ldots,b_n) =
\iiiSATC_{n+1}(b_1,\ldots,b_n,\False)\;.
\end{ldispl}%
In other words, for each $n \in \Nat$, $\iiiSATC_{n+1}$ can in essence
handle all inputs that $\iiiSATC_n$ can handle.
We will come back to this phenomenon in Section~\ref{sect-projective}.

\iiiSATC\ is meant to correspond to \iiiSAT.
Therefore, the following theorem does not come as a surprise.
We identify in this theorem the Boolean function family
$\iiiSATC = \indfam{\iiiSATC_n}{n \in \Nat}$ with the unique function
$\funct{\iiiSATC}{\seqof{\Bool}}{\Bool}$ such that for each
$n \in \Nat$, for each $w \in \Bool^n$, $\iiiSATC(w) = \iiiSATC_n(w)$.
\begin{theorem}
\label{theorem-3SATC-is-NP-compl}
$\iiiSATC$ is \NPT-complete.
\end{theorem}
\begin{proof}
$\iiiSATC$ is \NPT-complete iff \iiiSATC\ is in \NPT\ and \iiiSATC\ is
\NPT-hard.
Because \iiiSAT\ is \NPT-complete, it is sufficient to prove that
\iiiSATC\ is polynomial-time Karp reducible to \iiiSAT\ and
\iiiSAT\ is polynomial-time Karp reducible to \iiiSATC, respectively.
In the rest of the proof, $\alpha$ is defined as above.

\iiiSATC\ is polynomial-time Karp reducible to \iiiSAT:
Take the function $f$ from $\seqof{\Bool}$ to the set of all
\iiiCNF-formulas containing the variables $v_1,\ldots,v_k$ for some
$k \in \Nat$ that is defined by
$f(b_1,\ldots,b_n) =
 \LAND{i \in [1,\max \set{\ndisj(k) \where \ndisj(k) \leq n}]
   \;\mathrm{s.t.}\; b_i = \True}{}
  \LOR{}{} \alpha(i)$.
Then we have that
$\iiiSATC(b_1,\ldots,b_n) = \iiiSAT(f(b_1,\ldots,b_n))$.
It remains to show that $f$ is polyno\-mial-time computable.
To compute $f(b_1,\ldots,b_n)$, $\alpha$ has to be computed for a number
of times that is not greater than $n$ and $\alpha$ is computable in time
polynomial in $n$.
Hence, $f$ is polynomial-time computable.

\iiiSAT\ is polynomial-time Karp reducible to \iiiSATC:
Take the unique function $g$ from the set of all \iiiCNF-formulas
containing the variables $v_1,\ldots,v_k$ for some $k \in \Nat$ to
$\seqof{\Bool}$ such that for all \iiiCNF-formulas $\phi$ containing the
variables $v_1,\ldots,v_k$ for some $k \in \Nat$, $f(g(\phi)) = \phi$
and there exists no $w \in \seqof{\Bool}$ shorter than $g(\phi)$ such
that $f(w) = \phi$.
We have that $\iiiSAT(\phi) = \iiiSATC(g(\phi))$.
It remains to show that $g$ is polynomial-time computable.
Let $l$ be the size of $\phi$.
To compute $g(\phi)$, $\alpha$ has to be computed for each clause a
number of times that is not greater than $\ndisj(l)$ and $\alpha$ is
computable in time polynomial in $\ndisj(l)$.
Moreover, $\phi$ contains at most $l$ clauses.
Hence, $g$ is polynomial-time computable.
\qed
\end{proof}

Before we turn to the non-uniform super-polynomial complexity
conjecture, we touch lightly on the choice of the family of bijections
in the definition of $\iiiSATC$.
It is easy to see that the choice is not essential.
\pagebreak[2]
Let $\iiiSATC'$ be the same as \iiiSATC, but based on another family of
bijections, say $\indfam{\alpha'_n}{n \in \Nat}$, and let, for each
$i \in \Nat$, for each $j \in [1,\ndisj(i)]$,
$b'_j = b_{{\alpha_i}^{-1}(\alpha'_i(j))}$.
Then:
\begin{itemize}
\item
if $n = \ndisj(k)$ for some $k \in \Nat$:
\begin{ldispl}
\iiiSATC_n(b_1,\ldots,b_n) =
\iiiSATC'_n(b'_1,\ldots,b'_n)\;;
\end{ldispl}%
\item
if $\ndisj(k) < n < \ndisj(k+1)$ for some $k \in \Nat$:
\begin{ldispl}
\iiiSATC_n(b_1,\ldots,b_n) =
\iiiSATC'_n(b'_1,\ldots,b'_{\ndisj(k)},
                 b_{\ndisj(k)+1},\ldots,b_n)\;,
\end{ldispl}%
where $k$ is such that $\ndisj(k) < n < \ndisj(k+1)$.
\end{itemize}
This means that the only effect of another family of bijections is
another order of the relevant arguments.

The
\emph{non-uniform super-polynomial complexity conjecture} is the
following conjecture:
\begin{conjecture}
\label{conjecture-basic}
$\iiiSATC \notin \PLIS$.
\end{conjecture}

$\iiiSATC \notin \PLIS$ expresses in short that there does not exist a
polynomial function $\funct{h}{\Nat}{\Nat}$ such that for all
$n \in \Nat$ there exists an $X \in \ISbr$ such that $X$ computes
$\iiiSATC_n$ and $\psize(X) \leq h(n)$.
This corresponds with the following informal formulation of the
non-uniform super-polynomial complexity conjecture:
\begin{quote}
the lengths of the shortest instruction sequences that compute the \\
Boolean functions $\iiiSATC_n$ are not bounded by a polynomial in $n$.
\end{quote}

The statement that Conjecture~\ref{conjecture-basic} is a counterpart of
the conjecture that $\iiiSAT \notin \PTpoly$ is made rigorous in the
following theorem.
\begin{theorem}
\label{theorem-equiv-conjectures}
$\iiiSATC \notin \PLIS$ iff\, $\iiiSAT \notin \PTpoly$.
\end{theorem}
\begin{proof}
This follows immediately from Theorems~\ref{theorem-PLIS-is-PTpoly}
and~\ref{theorem-3SATC-is-NP-compl} and the fact that \iiiSAT\ is
\NPT-complete.
\qed
\end{proof}

\section{The Complexity Classes $\ndnuc{\IS}{\FN}$}
\label{sect-complclass-NPLIS}

In this section, we introduce a kind of non-uniform complexity classes
which includes a counterpart of the complexity class \NPTpoly\ in the 
setting of single-pass instruction sequences and show that this 
counterpart coincides with \NPTpoly.
Some results in this section are based on the definition of \NPT\ in 
terms of \PT, which can for example be found in~\cite{AB09a} (and which 
uses the idea of checking certificates), and the general definition of 
non-uniform complexity classes $\mathcal{C}/\FN$, which can for example 
be found in~\cite{BDG88a} (and which uses the idea of taking advice).

Let $\IS \subseteq \ISbr$ and 
let $\FN \subseteq \set{h \where \funct{h}{\Nat}{\Nat}}$.
Then $\ndnuc{\IS}{\FN}$ is the class of all Boolean function families
$\indfam{f_n}{n \in \Nat}$ that satisfy:
\begin{quote}
there exist a monotonic $h \in \FN$ and a Boolean function family
$\indfam{g_n}{n \in \Nat} \in \nuc{\IS}{\FN}$ such that, 
for all $n \in \Nat$, for all $w \in \Bool^n$:
\begin{ldispl}
f_n(w) = \True \Liff 
\Exists{c \in \Bool^{h(n)}}{g_{n+h(n)}(w\,c) = \True}\;.
\end{ldispl}%
\end{quote}
If a $c \in \seqof{\Bool}$ and a $w \in \Bool^n$ for which 
$f_n(w) = \True$ satisfy $g_{n+h(n)}(w\,c) = \True$, 
then we call $c$ a \emph{certificate} for $w$.
\pagebreak[2]

In the sequel, the monotonicity requirement in the definition given 
above is only used to show that \NPLIS\ coincides with \NPTpoly\ (see 
Theorems~\ref{theorem-char-NPLIS} and~\ref{theorem-NPLIS-is-NPTpoly}).

For each $\IS \subseteq \ISbr$ and 
$\FN \subseteq \set{h \where \funct{h}{\Nat}{\Nat}}$, the connection 
between the complexity classes $\nuc{\IS}{\FN}$ and $\ndnuc{\IS}{\FN}$ 
is like the connection between the complexity classes \PT\ and \NPT\ in 
the sense that it concerns the difference in complexity between finding 
a valid solution and checking whether a given solution is valid.

We are primarily interested in the complexity class $\NPLIS$,%
\footnote
{In precursors of this paper, the temporary name 
 $\mathrm{P}^{**}$ is used for the complexity class $\NPLIS$ 
 (see e.g.~\cite{BM08g}).}
but we will also pay attention to other instantiations of the general 
definition just given.

We have that \PLIS\ is included in \NPLIS.
\begin{theorem}
\label{theorem-PLIS-incl-NPLIS}
$\PLIS \subseteq \NPLIS$.
\end{theorem}
\begin{proof}
Suppose that $\indfam{f_n}{n \in \Nat} \in \PLIS$.
Then, for all $n \in \Nat$, for all $w \in \Bool^n$:
\begin{ldispl}
f_n(w) = \True \Liff
\Exists{c \in \Bool^{h(n)}}{f_{n+h(n)}(w\,c) = \True}
\end{ldispl}%
for the monotonic $h \in \poly$ defined by $h(n) = 0$ for all 
$n \in \Nat$, because the empty sequence can be taken as certificate for 
all $w$.
\qed
\end{proof}

Henceforth, we will use the notation $\ndnuc{\IS}{O(f(n))}$ for the 
complexity class
$\ndnuc{\IS}{\set{h \where \funct{h}{\Nat}{\Nat} \Land h(n) = O(f(n))}}$.
This notation is among other things used in the following corollary of 
the proof of Theorem~\ref{theorem-PLIS-incl-NPLIS}.
\begin{corollary}
\label{corollary-PLIS-incl-NPLIS}
For each $k \in \Nat$, 
$\nuc{\ISbr}{O(n^k)} \subseteq \ndnuc{\ISbr}{O(n^k)}$.
\end{corollary}

Henceforth, we will use the notation $\ndnuc{\IS}{B(f(n))}$ for the 
complexity class \sloppy
$\ndnuc{\IS}{\set{h \where \funct{h}{\Nat}{\Nat} \Land
                           \Forall{n \in \Natpos}{h(n) \leq f(n)}}}$.
This notation is among other things used in the following theorem about 
the complexity class $\ndnuc{\ISbr}{O(n^k)}$.
\begin{theorem}
\label{theorem-polynomials}
For each $k \in \Nat$, 
$\ndnuc{\ISbr}{O(n^k)} \subseteq
 \Union{a \in \Natpos} \ndnuc{\ISbr}{B(a\,n^k)}$.
\end{theorem}
\begin{proof}
Let $k \in \Nat$.
It is a direct consequence of the definition of $\ndnuc{\ISbr}{O(n^k)}$  
that, for all Boolean function families $\indfam{f_n}{n \in \Nat}$,
$\indfam{f_n}{n \in \Nat} \in \ndnuc{\ISbr}{O(n^k)}$ implies that
there exists an $a \in \Natpos$ such that
$\indfam{f_n}{n \in \Nat} \in \ndnuc{\ISbr}{B(a\,n^k)}$.
Hence 
$\ndnuc{\ISbr}{O(n^k)} \subseteq
 \Union{a \in \Natpos} \ndnuc{\ISbr}{B(a\,n^k)}$.
\qed
\end{proof}

In the proof of the following hierarchy theorem for \NPLIS, we will use 
the notation 
$\ndnucn{\IS}{\FN}{n}$ for  
$\set{\funct{f}{\Bool^n}{\Bool} \where 
      \Exists{\indfam{g_n}{n \in \Nat} \in \ndnuc{\IS}{\FN}}{f = g_n}}$.
\begin{theorem}
\label{theorem-hierarchy-NPLIS}
For each $k \in \Nat$, 
$\ndnuc{\ISbr}{O(n^k)} \subset \ndnuc{\ISbr}{O(n^{k^2+3})}$.
\end{theorem}
\begin{proof}
We will prove this theorem by defining a Boolean function family that by 
definition does not belong to $\ndnuc{\ISbr}{O(n^k)}$ and showing that 
it does belong to $\ndnuc{\ISbr}{O(n^{k^2+3})}$.
The definition concerned makes use of a natural number $m_a$ and an 
$m_a$-ary Boolean function $g_a$ for each positive natural number $a$.
These auxiliaries are defined first, using additional auxiliaries.
Let $k \in \Nat$.

For each $a \in \Natpos$, let the function $\funct{H_a}{\Nat}{\Nat}$ be 
defined by $H_a(n) = (3n + 10a(n + a\,n^k)^k - 2)^{a(n + a\,n^k)^k}$.
By Proposition~\ref{prop-inseq-number}, the remarks immediately 
preceding Proposition~\ref{prop-inseq-number} and the definition of 
$\ndnuc{\ISbr}{B(a\,n^k)}$, we have that, for each $a \in \Natpos$, 
the number of $n$-ary Boolean functions that 
belong to $\ndnucn{\ISbr}{B(a\,n^k)}{n}$ is not greater than $H_a(n)$ if
$n > 0$.
So $|\ndnucn{\ISbr}{B(a\,n^k)}{n}| \leq H_a(n)$ if $n > 0$.
By simple arithmetical calculations we find that 
$H_a(n) < (12(a + 1)^{k+1} n^{(k^2)})^{(a+1)^{k+1} n^{(k^2)}} 
        < 2^{4(a+1)^{2(k+1)} n^{k^2+1}}$ if $n > 0$.
Hence $H_a(n) < 2^{n^{k^2+2}}$ if $n \geq 4(a+1)^{2(k+1)}$.
By simple arithmetical calculations we also find that 
$(4(a+1)^{2(k+1)})^{k^2+2} = 2^{(k^2+2)(2(k+1)\log(a+1)+\log(4))}$ and
that
$(k^2 + 2)(2(k + 1)\log(a + 1) + \log(4)) < 2(k^2 + 2)(k + 1)(a + 1) <
 4(a + 1)^{2(k+1)}$.
Hence $n^{k^2+2} < 2^n$ if $n \geq 4(a+1)^{2(k+1)}$.
For each $a \in \Natpos$, let $m_a = 4(a+1)^{2(k+1)}$.
We immediately have that, for all $a \in \Natpos$, $m_a < m_{a+1}$ and,
for all $n \geq m_a$, $H_a(n) < 2^{n^{k^2+2}}$ and 
$n^{k^2+2} + 1 \leq 2^n$.

For each $a \in \Natpos$, 
let $S_a \in (\Bool^{m_a})^{{m_a}^{k^2+2} + 1}$ be such that 
the elements of $S_a$ are mutually different, and 
let $\widehat{S_a}$ be the set of all elements of $S_a$.
Because ${m_a}^{k^2+2} + 1 \leq 2^{m_a}$, we have that 
$\widehat{S_a} \subseteq \Bool^{m_a}$.
The number of functions from $\widehat{S_a}$ to $\Bool$ is
$2^{{m_a}^{k^2+2} + 1}$, and we have that 
$2^{{m_a}^{k^2+2} + 1} >
 H_a(m_a) \geq |\ndnucn{\ISbr}{B(a\,{m_a}^k)}{m_a}|$.
Because each function $\funct{f'}{\widehat{S_a}}{\Bool}$ can be extended 
to a function $\funct{f}{\Bool^{m_a}}{\Bool}$ by defining
$f(w) = f'(w)$ if $w \in \widehat{S_a}$ and $f(w) = \False$ otherwise,
this means that there exists an $m_a$-ary Boolean function that does not
belong to $\ndnucn{\ISbr}{B(a\,{m_a}^k)}{m_a}$.
For each $a \in \Natpos$, let $g_a$ be such an $m_a$-ary Boolean function.

Let $\indfam{g_n}{n \in \Nat}$ be the Boolean function family such that, 
for each $n \in \Nat$, $g_n$ is defined as follows:
\begin{itemize}
\item
if $n = m_a$ for some $a \in \Natpos$: $g_n = g_a$, where $a$ is such 
that $n = m_a$;
\item
if $n \neq m_a$ for all $a \in \Natpos$: $g_n(w) = \False$ for all 
$w \in \Bool^n$.
\end{itemize}
For each $a \in \Natpos$, there exists an $n \in \Natpos$ such that $g_n$ 
does not belong to $\ndnucn{\ISbr}{B(a\,n^k)}{n}$.
Hence, by Theorem~\ref{theorem-polynomials}, 
$\indfam{g_n}{n \in \Nat} \notin \ndnuc{\ISbr}{O(n^k)}$.

For each $n \in \Nat$, we can construct an instruction sequence $X$ that
computes $g_n$ as follows:
\begin{itemize}
\item
if $n = m_a$ for some $a \in \Natpos$: 
$X = X^{w_1,g(w_1)} \conc \ldots \conc
     X^{w_{n^{k^2+2} + 1},g(w_{n^{k^2+2} + 1})} \conc \halt$,
where for each $i \in [1,n^{k^2+2} + 1]$, $w_i$ is the $i$th element of 
$S_a$ (where $a$ is such that $n = m_a$), and
$X^{w,b}$ is an instruction sequence that sets the output register to 
$b$ if all input registers together contain $w$ and jumps to the next 
instruction sequence otherwise;
\item
if $n \neq m_a$ for all $a \in \Natpos$: 
$X = \halt$.
\end{itemize}
In the case where $n = m_a$ for some $a \in \Natpos$, 
$\psize(X) \leq (n^{k^2+2} + 1) \cdot (2 n + 1) + 1$.
Otherwise, $\psize(X) = 1$.
Hence, $\indfam{g_n}{n \in \Nat} \in \nuc{\ISbr}{O(n^{k^2+3})}$.
From this and Corollary~\ref{corollary-PLIS-incl-NPLIS}, it follows that
$\indfam{g_n}{n \in \Nat} \in \ndnuc{\ISbr}{O(n^{k^2+3})}$.
\qed
\end{proof}
The approach followed in the proof of the hierarchy theorem for $\PLIS$
does not seem to work for the proof of the hierarchy theorem for 
$\NPLIS$.

In the general definition of the complexity classes $\ndnuc{\IS}{\FN}$, 
a pair of an input and a certificate is turned into a single sequence by 
simply concatenating the input and the certificate.
In the usual definition of \NPT\ in terms of \PT, on the other hand, a 
pair of an input and a certificate is uniquely encoded by a single 
sequence from which the input and certificate are recoverable.
A function that does so is commonly called a pairing function.
In the case of \NPLIS, a definition in which a pair of an input and a 
certificate is turned into a single sequence in the latter way could 
have been given as well.

Consider the pairing function from $\seqof{\Bool} \x \seqof{\Bool}$ to 
$\seqof{\Bool}$ that converts each two sequences 
$(b_1,\ldots,b_n)$ and $(b'_1,\ldots,b'_{n'})$ into 
$(b_1,b_1,\ldots,b_n,b_n,\True,\False,b'_1,\ldots,b'_{n'})$.
Henceforth, we will write $w \pairing w'$, 
where $w,w' \in \seqof{\Bool}$, to denote the result of applying this 
pairing function to $w$ and $w'$.
Take, for each $m \in \Nat$, the function from 
$\Union{i \geq m} \Bool^i$ to $\seqof{\Bool}$ that converts each
sequence $(b_1,\ldots,b_m,b_{m+1},\ldots,b_n)$ into 
$(b_1,\ldots,b_m) \pairing (b_{m+1},\ldots,b_n)$, and
the two projection functions from the range of\linebreak[2] 
the above-mentioned 
pairing function to $\seqof{\Bool}$ that extract from each sequence 
$(b_1,\ldots,b_m) \pairing (b_{m+1},\ldots,b_{n})$ the sequences
$(b_1,\ldots,b_m)$ and $(b_{m+1},\ldots,b_{n})$.
Then, for all $n \in \Nat$, the restrictions of these functions to
the sequences that belong to $\Bool^n$ are computable by an instruction 
sequence $X \in \ISbrna$ with $\psize(X) = O(n)$.

The following theorem gives an alternative characterization of 
$\NPLIS$.
\begin{theorem}
\label{theorem-char-NPLIS}
Let $\indfam{f_n}{n \in \Nat}$ be a Boolean function family.
Then we have that $\indfam{f_n}{n \in \Nat} \in \NPLIS$ iff
\begin{quote}
there exist an $h \in \poly$ and a Boolean function family
\sloppy
$\indfam{g_n}{n \in \Nat} \in \PLIS$ such that,
for all $n \in \Nat$, for all $w \in \Bool^n$:
\begin{ldispl}
f_n(w) = \True \Liff
\Exists{c \in \seqof{\Bool}}
 {(|c| \leq h(n) \Land g_{|w \pairing c|}(w \pairing c) = \True)}\;.
\end{ldispl}%
\end{quote}
\end{theorem}
\begin{proof}
The implication from left to right follows directly from the definition 
of \NPLIS\ and the remark made above about the projection functions 
associated with the pairing function used here.

The implication from right to left is proved as follows.
Let $\indfam{g_n}{n \in \Nat} \in \PLIS$ be such that
there exists an $h \in \poly$ such that, 
for all $n \in \Nat$, for all $w \in \Bool^n$,  
$f_n(w) = \True \Liff 
 \Exists{c \in \seqof{\Bool}}
  {(|c| \leq h(n) \Land  
    g_{|w \pairing c|}(w \pairing c) = \True)}$.
For each $w \in \seqof{\Bool}$, let $c_w$ be a certificate for $w$.
Suppose that $\indfam{\gamma_n}{n \in \Nat}$, 
with $\funct{\gamma_n}{\Bool^n}{\Bool^*}$ for every $n \in \Nat$, 
is an infinite sequence of injective functions satisfying: 
(i)~there exists a monotonic $h \in \poly$ such that, for all 
$n \in \Nat$, for all $w \in \Bool^n$, $|\gamma_n(c_w)| = h(n)$,
(ii)~for all $n \in \Nat$, ${\gamma_n}^{-1}$ is computable by an 
instruction sequence $X \in \ISbrna$ with $\psize(X) = O(h(n))$.
Then there exists a monotonic $h \in \poly$ such that, 
for all $n \in \Nat$, for all $w \in \Bool^n$, 
$\gamma_n(c_w) \in \Bool^{h(n)}$ and
there exists a $\indfam{g'_n}{n \in \Nat} \in \PLIS$ such that, 
for all $n \in \Nat$, for all $w \in \Bool^n$, 
$g_{|w \pairing c|}(w \pairing c_w) = \True \Liff
 g'_{n+h(n)}(w\,\gamma_n(c_w)) = \True$.
The existence of a suitable $\indfam{g'_n}{n \in \Nat} \in \PLIS$ is 
not guaranteed for an $h \in \poly$ with the property that there exist 
$n,n' \in \Nat$ such that $n + h(n) = n' + h(n')$ and $n \neq n'$, but
this property is excluded by the required monotonicity of $h$.
It remains to show that the functions $\gamma_n$ supposed above 
exist. 
For $\gamma_n$, we can pick the function that converts each sequence
$(b_1,\ldots,b_n)$ into $(b_1,b'_1,\ldots,b_n,b'_n)$, where 
$b'_i = \True$ if $i \neq n$ and $b'_n = \False$, and adds at the 
end of the converted sequence sufficiently many $\False$'s to obtain 
results of the required length.
\qed
\end{proof}

It happens that \NPLIS\ and \NPTpoly\ coincide, provided that we 
identify each Boolean function family $\indfam{f_n}{n \in \Nat}$ with 
the unique function $\funct{f}{\seqof{\Bool}}{\Bool}$ such that for 
each $n \in \Nat$, for each $w \in \Bool^n$, $f(w) = f_n(w)$.
\begin{theorem}
\label{theorem-NPLIS-is-NPTpoly}
$\NPLIS = \NPTpoly$.
\end{theorem}
\begin{proof}
It follows by elementary reasoning from the general definition of 
non-uniform complexity classes $\mathcal{C}/\FN$ and the definition of 
$\NPT$ in terms of $\PT$ that $f \in \NPTpoly$ iff 
there exist a polynomial function $\funct{h}{\Nat}{\Nat}$ and a 
$g \in \PTpoly$ such that, for all $w \in \seqof{\Bool}$:
\begin{ldispl}
f(w) = \True \Liff
\Exists{c \in \seqof{\Bool}}
 {(|c| \leq h(|w|) \Land g(w \pairing c) = \True)}
\end{ldispl}%
(cf.\ Fact~2 in~\cite{Yap83a}).
From this characterization of \NPTpoly\ and the characterization of 
\NPLIS\ given in Theorem~\ref{theorem-char-NPLIS}, it follows easily 
that $\NPLIS = \NPTpoly$.
\qed
\end{proof}

In Section~\ref{sect-conjecture}, we have conjectured that 
$\iiiSATC \notin \PLIS$.
The question arises whether $\iiiSATC \in \NPLIS$.
This question can be answered in the affirmative.

\begin{theorem}
\label{theorem-3SATC-in-NPLIS}
$\iiiSATC \in \NPLIS$.
\end{theorem}
\begin{proof}
$\iiiSATC \in \NPT$ by Theorem~\ref{theorem-3SATC-is-NP-compl}, 
$\NPT \subseteq \NPTpoly$ by the general definition of non-uniform 
complexity classes $\mathcal{C}/\FN$ (see e.g.~\cite{BDG88a}), and 
$\NPTpoly = \NPLIS$ by Theorem~\ref{theorem-NPLIS-is-NPTpoly}.
Hence, $\iiiSATC \in \NPLIS$.
\qed
\end{proof}

\section{Completeness for the Complexity Classes $\ndnuc{\IS}{\FN}$}
\label{sect-NPLIS-compl}

In this section, we introduce the notion of 
$\ndnuc{\IS}{\FN}$-completeness, a general notion of completeness for 
complexity classes $\ndnuc{\IS}{\FN}$ where $\FN$ is closed under 
function composition. 
Like \NPT-completeness, $\ndnuc{\IS}{\FN}$-completeness will be 
defined in terms of a reducibility relation.

Let $\IS \subseteq \ISbr$, let $l,m,n \in \Nat$, and 
let $\funct{f}{\Bool^n}{\Bool}$ and $\funct{g}{\Bool^m}{\Bool}$.
Then $f$ is \emph{$l$-length $\IS$-reducible} to $g$, written 
$f \llred{\IS}{l} g$,
if there exist $\funct{h_1,\ldots,h_m}{\Bool^n}{\Bool}$ such that:
\begin{itemize}
\item
                                                  
there exist $X_1,\ldots,X_m \in \IS$ such that $X_1,\ldots,X_m$
compute $h_1,\ldots,h_m$ and $\psize(X_1),\ldots,\psize(X_m) \leq l$;
\item
for all $b_1,\ldots,b_n \in \Bool$,
$f(b_1,\ldots,b_n) = g(h_1(b_1,\ldots,b_n),\ldots,h_m(b_1,\ldots,b_n))$.
\end{itemize}
Let $\IS \subseteq \ISbr$, 
let $\FN \subseteq \set{h \where \funct{h}{\Nat}{\Nat}}$ be such that
$\FN$ is closed under func\-tion composition, and
let $\indfam{f_n}{n \in \Nat}$ and $\indfam{g_n}{n \in \Nat}$ be
Boolean function families. 
Then $\indfam{f_n}{n \in \Nat}$ is
\emph{non-uniform $\FN$-length $\IS$-reducible} to 
$\indfam{g_n}{n \in \Nat}$,
written \linebreak[2]
$\indfam{f_n}{n \in \Nat} \plred{\IS}{\FN} \indfam{g_n}{n \in \Nat}$,
if there exists an $h \in \FN$ such that:
\begin{itemize}
\item
for all $n \in \Nat$, there exist $l,m \in \Nat$ with
$l,m \leq h(n)$ such that $f_n \llred{\IS}{l} g_m$.
\end{itemize}

Let $\IS \subseteq \ISbr$, let $\FN$ be as above, and
let $\indfam{f_n}{n \in \Nat}$ be a Boolean function family.
Then $\indfam{f_n}{n \in \Nat}$ is
\emph{$\ndnuc{\IS}{\FN}$-com\-plete} if:
\begin{itemize}
\item
$\indfam{f_n}{n \in \Nat} \in\ndnuc{\IS}{\FN}$;
\item
for all $\indfam{g_n}{n \in \Nat} \in \ndnuc{\IS}{\FN}$,
$\indfam{g_n}{n \in \Nat} \plred{\IS}{\FN} \indfam{f_n}{n \in \Nat}$.
\end{itemize}

The most important properties of non-uniform $\FN$-length
$\IS$-reducibility and $\ndnuc{\IS}{\FN}$-completeness as defined above 
are stated in the following two propositions.
\begin{proposition}
\label{prop-plred-props}
Let $\IS \subseteq \ISbr$, and let $\FN$ be as above.
Then:
\begin{enumerate}
\item
if $\indfam{f_n}{n \in \Nat} \plred{\IS}{\FN} \indfam{g_n}{n \in \Nat}$ 
and $\indfam{g_n}{n \in \Nat} \in \nuc{\IS}{\FN}$,
then $\indfam{f_n}{n \in \Nat} \in \nuc{\IS}{\FN}$;
\item
$\plred{\IS}{\FN}$ is reflexive and transitive.
\end{enumerate}
\end{proposition}
\begin{proof}
Both properties follow immediately from the definition of 
$\plred{\IS}{\FN}$.
\qed
\end{proof}
\begin{proposition}
\label{prop-NPLIS-compl-props}
Let $\IS \subseteq \ISbr$, and let $\FN$ be as above.
Then:
\begin{enumerate}
\item
if $\indfam{f_n}{n \in \Nat}$ is $\ndnuc{\IS}{\FN}$-complete and
$\indfam{f_n}{n \in \Nat} \in \nuc{\IS}{\FN}$, then 
$\ndnuc{\IS}{\FN} = \nuc{\IS}{\FN}$;
\item
if $\indfam{f_n}{n \in \Nat}$ is $\ndnuc{\IS}{\FN}$-complete,
$\indfam{g_n}{n \in \Nat} \in \ndnuc{\IS}{\FN}$ and
$\indfam{f_n}{n \in \Nat} \plred{\IS}{\FN} \indfam{g_n}{n \in \Nat}$,
then $\indfam{g_n}{n \in \Nat}$ is $\ndnuc{\IS}{\FN}$-complete.
\end{enumerate}
\end{proposition}
\begin{proof}
The first property follows immediately from the definition of
$\ndnuc{\IS}{\FN}$-com\-pleteness, and the second property follows 
immediately from the definition of $\ndnuc{\IS}{\FN}$-completeness and 
the transitivity of $\plred{\IS}{\FN}$.
\qed
\end{proof}
The properties stated in Proposition~\ref{prop-NPLIS-compl-props} make
$\ndnuc{\IS}{\FN}$-completeness as defined above adequate for our 
purposes.
In the following proposition, non-uniform polynomial-length 
$\ISbr$-reducibility ($\plred{\ISbr}{\poly}$) is related to 
polynomial-time Karp reducibility ($\ptred$).%
\footnote
{For a definition of polynomial-time Karp reducibility, see e.g.\ 
 Chapter~2 of~\cite{AB09a}.}
\begin{proposition}
\label{prop-plred-psred}
Let $\indfam{f_n}{n \in \Nat}$ and $\indfam{g_n}{n \in \Nat}$ be
Boolean function families, and let $f$ and $g$ be the unique functions
$\funct{f,g}{\seqof{\Bool}}{\Bool}$ such that for each $n \in \Nat$, 
\linebreak[2]
for each $w \in \Bool^n$, $f(w) = f_n(w)$ and $g(w) = g_n(w)$.
Then $f \ptred g$ only if 
\mbox{$\indfam{f_n}{n \in \Nat} \plred{\ISbr}{\poly} 
       \indfam{g_n}{n \in \Nat}$.}
\end{proposition}
\begin{proof}
This property follows immediately from the definitions of $\ptred$ 
and \smash{$\plred{\ISbr}{\poly}$}, the fact that 
$\PT \subseteq \PTpoly$ (which follows directly from the general 
definition of non-uniform complexity classes $\mathcal{C}/\FN$), and 
Theorem~\ref{theorem-PLIS-is-PTpoly}.
\qed
\end{proof}
The property stated in Proposition~\ref{prop-plred-psred} allows for
results concerning polynomial-time Karp reducibility to be reused when 
dealing with non-uniform polynomial-length $\ISbr$-reducibility.

We would like to call \NPLIS-completeness the counterpart of
\NPTpoly-completeness in the current setting, but the notion of
\NPTpoly-completeness looks to be absent in the literature on
complexity theory.
The closest to \NPTpoly-completeness that we could find is
$p$-completeness for pD, a notion introduced in~\cite{SV85a}.

Because \iiiSATC\ is closely related to \iiiSAT\ and
$\iiiSATC \in \NPLIS$, we expect \iiiSATC\ to be \NPLIS-complete.
\begin{theorem}
\label{theorem-3SATC-is-NPLIS-compl}
\iiiSATC\ is \NPLIS-complete.
\end{theorem}
\begin{proof}
By Theorem~\ref{theorem-3SATC-in-NPLIS}, we have that
$\iiiSATC \in \NPLIS$.
It remains to prove that for all $\indfam{f_n}{n \in \Nat} \in \NPLIS$,
$\indfam{f_n}{n \in \Nat} \plred{\ISbr}{\poly} \iiiSATC$.

\pagebreak[2]
Suppose that $\indfam{f_n}{n \in \Nat} \in \NPLIS$. 
Let $\indfam{g_n}{n \in \Nat} \in \PLIS$ be such that 
there exists a monotonic $h \in \poly$ such that, 
for each $n \in \Nat$, for each $w \in \Bool^n$,
$f_n(w) = \True \Liff
 \Exists{c \in \Bool^{h(n)}}{g_{n+h(n)}(w\,c) = \True}$.
Such a $\indfam{g_n}{n \in \Nat}$ exists by the definition of \NPLIS.
Let $h \in \poly$ be such that, for each $n \in \Nat$, for each 
$w \in \Bool^n$,
$f_n(w) =\nolinebreak \True \Liff 
 \Exists{c \in \Bool^{h(n)}}{g_{n+h(n)}(w\,c) = \True}$.
Let $n \in \Nat$, and let $m = h(n)$.
Let $X \in \ISbr$ be such that $X$ computes $g_{n+m}$ and 
$\psize(X)$ is polynomial in $n + m$.

Assume that $\outbr.\setbr{\True}$ occurs only once in $X$, that 
$\fjmp{l}$ does not occur in $X$ at positions where there is no $l$th 
next primitive instruction, and that test instructions do not occur in 
$X$ at the last but one or last position.
These assumptions can be made without loss of generality: 
by Theorem~\ref{theorem-comput-output} we may assume without loss of 
generality that $\outbr.\setbr{\False}$ does not occur in $X$ and 
therefore multiple occurrences of $\outbr.\setbr{\True}$ can always be 
eliminated by replacing them with the exception of the last one by jump 
instructions, occurrences of $\fjmp{l}$ at positions where there is no 
$l$th next primitive instruction can always be eliminated by replacing 
them by $\fjmp{0}$, and occurrences of test instructions at the last but 
one or last position can be eliminated by adding once or twice 
$\fjmp{0}$ at the end.
Suppose that $X = u_1 \conc \ldots \conc u_k$, and let $l \in [1,k]$ be
such that $u_l$ is either $\outbr.\setbr{\True}$,
$\ptst{\outbr.\setbr{\True}}$ or $\ntst{\outbr.\setbr{\True}}$.

First of all, we look for a transformation that gives, for each
$b_1,\ldots,b_n \in \Bool$, a Boolean formula $\phi_{b_1,\ldots,b_n}$
such that $f_n(b_1,\ldots,b_n) = \True$ iff $\phi_{b_1,\ldots,b_n}$ is
satisfiable.
We have that $f_n(b_1,\ldots,b_n) = \True$ iff there exist initial 
states of the Boolean registers named $\inbr{n{+}1},\ldots,\inbr{n{+}m}$ 
for which there exists an execution path through $X$ that reaches $u_l$ 
in case the initial states of the Boolean registers named 
$\inbr{1},\ldots,\inbr{n}$ are $b_1,\ldots,b_n$, respectively.
This brings up the formula $\phi_{b_1,\ldots,b_n}$ given below.
In this formula, propositional variables $r_1,\ldots,r_{n+m}$ and 
$v_1,\ldots,v_k$ are used.
The truth value assigned to $r_i$ ($i \in [1,n+m]$) is intended to 
indicate whether the content of the input register named $\inbr{i}$ is 
$\True$ and the truth value assigned to $v_j$ ($j \in [1,k]$) is 
intended to indicate whether the primitive instruction $u_i$ is 
executed.

For each $b_1,\ldots,b_n \in \Bool$, let $\phi_{b_1,\ldots,b_n}$ be 
$\LAND{i \in [1,n]}{} \chi_i \Land v_1 \Land v_l \Land
 \LAND{j \in [1,k]}{} \psi_j$, where: 
for each $i \in [1,n]$, 
$\chi_i$ is 
\begin{itemize}
\item
$r_i \phantom{\Lnot {}}$ if $b_i = \True$;
\item
$\Lnot r_i$ if $b_i = \False$; 
\end{itemize}
for each $j \in [1,k]$, $\psi_j$ is 
\begin{itemize}
\item
$v_j \Limpl v_{j+1}$ if $u_j \equiv a$;
\item
$v_j \Limpl v_{j+1} \phantom{{} \Land \Lnot v_{j+1}}$ 
if $u_j \equiv \ptst{f.\setbr{\True}}$ or 
   $u_j \equiv \ntst{f.\setbr{\False}}$;
\item
$v_j \Limpl \Lnot v_{j+1} \Land v_{j+2}$ 
if $u_j \equiv \ptst{f.\setbr{\False}}$ or 
   $u_j \equiv \ntst{f.\setbr{\True}}$;
\item
$(v_j \Land r_i \Limpl v_{j+1}) \Land
 (v_j \Land \Lnot r_i \Limpl \Lnot v_{j+1} \Land v_{j+2})$ 
if $u_j \equiv \ptst{\inbr{i}.\getbr}$;
\item
$(v_j \Land \Lnot r_i \Limpl v_{j+1}) \Land
 (v_j \Land r_i \Limpl \Lnot v_{j+1} \Land v_{j+2})$ 
if $u_j \equiv \ntst{\inbr{i}.\getbr}$;
\item
$(v_j \Land \LOR{j' \in B^{i,\True}_{1,j-1}}{\!} 
   (v_{j'} \Land \LAND{j'' \in B^{i,\False}_{j'+1,j-1}}{\!} \Lnot v_{j''})
     \Limpl v_{j+1}) \Land {}$ \\
$(v_j \Land \LOR{j' \in B^{i,\False}_{1,j-1}}{\!} 
   (v_{j'} \Land \LAND{j'' \in B^{i,\True}_{j'+1,j-1}}{\!} \Lnot v_{j''})
     \Limpl \Lnot v_{j+1} \Land v_{j+2})$ \\
if $u_j \equiv \ptst{\auxbr{i}.\getbr}$;
\item
$(v_j \Land \LOR{j' \in B^{i,\False}_{1,j-1}}{\!} 
   (v_{j'} \Land \LAND{j'' \in B^{i,\True}_{j'+1,j-1}}{\!} \Lnot v_{j''})
     \Limpl v_{j+1}) \Land {}$ \\
$(v_j \Land \LOR{j' \in B^{i,\True}_{1,j-1}}{\!} 
   (v_{j'} \Land \LAND{j'' \in B^{i,\False}_{j'+1,j-1}}{\!} \Lnot v_{j''})
     \Limpl \Lnot v_{j+1} \Land v_{j+2})$ \\
if $u_j \equiv \ntst{\auxbr{i}.\getbr}$;
\item
$\Lnot v_j 
 \phantom{v_{j+1} \Limpl \LAND{j' \in [j+l,j+l-1]}{} v_{j'} \Land {}}$ 
if $u_j \equiv \fjmp{0}$;
\item
$v_j \Limpl 
 \LAND{j' \in [j+1,j+l-1]}{\!} \Lnot v_{j'} \Land v_{j+l}$
if $u_j \equiv \fjmp{l}$ and $1 \leq l \leq k-j $;
\item
$v_j$ if $u_j \equiv \halt$; $\phantom{\LAND{j \in [1,l]}{}}$
\end{itemize}
where $B^{i,b}_{j,j'}$ is the set of all $j'' \in [j,j']$ for which 
$u_{j''}$ is either $\auxbr{i}.\setbr{b}$, $\ptst{\auxbr{i}.\setbr{b}}$ 
or $\ntst{\auxbr{i}.\setbr{b}}$.

If there exist initial states of the Boolean registers named 
$\inbr{n{+}1},\ldots,\inbr{n{+}m}$ for which there exists an execution 
path through $X$ that reaches $u_l$ in case the initial states of the 
Boolean registers named $\inbr{1},\ldots,\inbr{n}$ are $b_1,\ldots,b_n$, 
respectively, then $\phi_{b_1,\ldots,b_n}$ is satisfiable by assigning
truth values to the variables according to the intention mentioned 
above.
On the other hand, if $\phi_{b_1,\ldots,b_n}$ is satisfiable, then a
satisfying assignment indicates for which initial states of the Boolean 
registers named $\inbr{n{+}1},\ldots,\inbr{n{+}m}$ there exists an 
execution path through $X$ that reaches $u_l$ and which instructions are 
on this execution path.
Thus $f_n(b_1,\ldots,b_n) = \True$ iff $\phi_{b_1,\ldots,b_n}$ is
satisfiable.

For some $l \in \Nat$, $\phi_{b_1,\ldots,b_n}$ still has to be
transformed into a $w_{b_1,\ldots,b_n} \in \Bool^l$ such that
$\phi_{b_1,\ldots,b_n}$ is satisfiable iff
$\iiiSATC_l(w_{b_1,\ldots,b_n}) = \True$.
We look upon this transformation as a composition of two
transformations: first $\phi_{b_1,\ldots,b_n}$ is transformed into a
\iiiCNF-formula $\psi_{b_1,\ldots,b_n}$ such that
$\phi_{b_1,\ldots,b_n}$ is satisfiable iff $\psi_{b_1,\ldots,b_n}$ is
satisfiable, and next, for some $l \in \Nat$, $\psi_{b_1,\ldots,b_n}$ is
transformed into a $w_{b_1,\ldots,b_n} \in \Bool^l$ such that
$\psi_{b_1,\ldots,b_n}$ is satisfiable iff
$\iiiSATC_l(w_{b_1,\ldots,b_n}) = \True$.

It is easy to see that the size of $\phi_{b_1,\ldots,b_n}$ is polynomial
in $n$ and that \sloppy $\tup{b_1,\ldots,b_n}$ can be transformed into
$\phi_{b_1,\ldots,b_n}$ in time polynomial in $n$.
It is well-known that each Boolean formula $\psi$ can be transformed in
time polynomial in the size of $\psi$ into a \iiiCNF-formula $\psi'$,
with size and number of variables linear in the size of $\psi$, such
that $\psi$ is satisfiable iff $\psi'$ is satisfiable
(see e.g.\ Theorem~3.7 in~\cite{BDG88a}).
Moreover, it is known from the proof of
Theorem~\ref{theorem-3SATC-is-NP-compl} that every \iiiCNF-formula 
$\phi$ can be transformed in time polynomial in the size of $\phi$ into 
a $w \in \Bool^{\ndisj(k')}$, where $k'$ is the number of variables in
$\phi$, such that $\iiiSAT(\phi) = \iiiSATC(w)$.
From these facts, and Proposition~\ref{prop-plred-psred}, it follows
easily that $\indfam{f_n}{n \in \Nat}$ is non-uniform polynomial-length
$\ISbr$-reducible to $\iiiSATC$.
\qed
\end{proof}
The proof of Theorem~\ref{theorem-3SATC-is-NPLIS-compl} has been partly
inspired by the proof of the NP-completeness of SAT in~\cite{Coo71a}.

A known result about classical complexity classes turns out to be a
corollary of
Theorems~\ref{theorem-PLIS-is-PTpoly}, \ref{theorem-3SATC-is-NP-compl},
\ref{theorem-NPLIS-is-NPTpoly} and~\ref{theorem-3SATC-is-NPLIS-compl}.
\begin{corollary}
\label{corollary-equiv-conjectures}
$\NPT \not\subseteq \PTpoly$ iff\, $\NPTpoly \not\subseteq \PTpoly$.
\end{corollary}

\section{Projective Boolean Function Families}
\label{sect-projective}

In Section~\ref{sect-conjecture}, we have noticed that, for each 
$n \in \Nat$, $\iiiSATC_{n+1}$ can in essence handle all inputs that 
$\iiiSATC_n$ can handle because we have 
$\iiiSATC_n(b_1,\ldots,b_n) = \iiiSATC_{n+1}(b_1,\ldots,b_n,\False)$.
In this section, we come back to this phenomenon.

For each $m,n \in \Nat$ such that $m \geq n$, we define a 
\emph{projection} function 
$\funct{\pi^m_n}{(\Bool^m \to \Bool)}{(\Bool^n \to \Bool)}$ as follows:
\begin{ldispl}
\pi^m_n(f)(b_1,\ldots,b_n) = 
\smash
{f(b_1,\ldots,b_n,
   \overbrace{\False,\ldots,\False}^{m - n\; \x})}
\end{ldispl}
for all $\funct{f}{\Bool^m}{\Bool}$ and $b_1,\ldots,b_n \in \Bool$.

A \emph{projective Boolean function family} is a Boolean function family 
$\indfam{f_n}{n \in \Nat}$ such that $f_n = \pi^{n+1}_n(f_{n+1})$ for 
all $n \in \Nat$.

This means that in the case where a Boolean function family
$\indfam{f_n}{n \in \Nat}$ is projective, for each $m,n \in \Nat$ with 
$m > n$, $f_m$ can in essence handle all inputs that $f_n$ can handle.
For that reason, complexity classes that are restricted to projective 
Boolean function families are potentially interesting. 

Let $\IS \subseteq \ISbr$ and 
$\FN \subseteq \set{h \where \funct{h}{\Nat}{\Nat}}$.
Then $\pnuc{\IS}{\FN}$ is the class of all projective Boolean function 
families $\indfam{f_n}{n \in \Nat}$ that satisfy:
\begin{quote}
there exists an $h \in \FN$ such that
for all $n \in \Nat$ there exists an $X \in \IS$ such that
$X$ computes $f_n$ and $\psize(X) \leq h(n)$.
\end{quote}

Let $\IS \subseteq \ISbr$ and 
let $\FN \subseteq \set{h \where \funct{h}{\Nat}{\Nat}}$.
Then $\ndpnuc{\IS}{\FN}$ is the class of all projective Boolean function 
families $\indfam{f_n}{n \in \Nat}$ that satisfy:
\begin{quote}
there exist a monotonic $h \in \FN$ and a Boolean function family
$\indfam{g_n}{n \in \Nat} \in \pnuc{\IS}{\FN}$ such that, 
for all $n \in \Nat$, for all $w \in \Bool^n$:
\begin{ldispl}
f_n(w) = \True \Liff 
\Exists{c \in \Bool^{h(n)}}{g_{n+h(n)}(w\,c) = \True}\;.
\end{ldispl}%
\end{quote}

It follows immediately from the definitions concerned that 
$\pnuc{\IS}{\FN}$ and $\ndpnuc{\IS}{\FN}$ are subsets of 
$\nuc{\IS}{\FN}$ and $\ndnuc{\IS}{\FN}$, respectively. 
$\pnuc{\ISbr\hsp{-.1}}{\poly}$ is a proper subset of \PLIS.
This is easy to see.
In Section~\ref{sect-complclass-PLIS}, we gave an example of a Boolean 
function family corresponding to an uncomputable function from 
$\seqof{\Bool}$ to $\Bool$ that belongs to \PLIS.
The Boolean function family concerned is not a projective Boolean 
function family, and consequently does not belong to 
$\pnuc{\ISbr\hsp{-.1}}{\poly}$.

The question arises whether the restriction to projective Boolean 
function families is a severe restriction.
We have that every Boolean function family is non-uniform linear-length 
$\ISbr$-reducible to a projective Boolean function family.

Below, we will write $\lin$ for the set 
$\set{h \where \funct{h}{\Nat}{\Nat} \Land h \mathrm{\,is \,linear}}$.

\begin{theorem}
\label{theorem-projective}
Let $\indfam{f_n}{n \in \Nat}$ be a Boolean function family.
Then there exists a projective Boolean function family 
$\indfam{g_n}{n \in \Nat}$ such that 
$\indfam{f_n}{n \in \Nat} \plred{\ISbr}{\lin}
 \indfam{g_n}{n \in \Nat}$.
\end{theorem}
\begin{proof}
The idea is to convert inputs $(b_1,\ldots,b_n)$ into
$(b_1,b'_1,\ldots,b_n,b'_n)$, where $b'_i = \True$ if $i \neq n$ and
$b'_n = \False$, because the converted inputs can be recovered after 
additions at the end.
This conversion has been used before to prove 
Theorem~\ref{theorem-char-NPLIS}.

For each $n \in \Nat$, for each $i \in [1,n]$, 
let $\funct{h^n_{2i-1}}{\Bool^n}{\Bool}$ be defined by 
$h^n_{2i-1}(b_1,\ldots,b_n) = b_i$ and 
let $\funct{h^n_{2i}}{\Bool^n}{\Bool}$   be defined by 
$h^n_{2i}(b_1,\ldots,b_n) = \True$  if $i \neq n$ and
$h^n_{2i}(b_1,\ldots,b_n) = \False$ if $i = n$.
Clearly, these functions can be computed by instruction sequences from
$\ISbr$ whose lengths are linear in $n$.
Therefore, we are done with the proof if we show that there exists a 
projective Boolean function family $\indfam{g_n}{n \in \Nat}$ such that 
for all $n \in \Nat$:
\begin{ldispl}
f_n(b_1,\ldots,b_n) = 
g_{2n}(h^n_1(b_1,\ldots,b_n),\ldots,h^n_{2n}(b_1,\ldots,b_n))\;.
\end{ldispl}
A witness is the projective Boolean function family 
$\indfam{g_n}{n \in \Nat}$ with, for each $n \in \Nat$, 
$\funct{g_{2n}}{\Bool^{2n}}{\Bool}$ defined by 
$g_{2n}(b_1,b'_1,\ldots,b_n,b'_n) = f_m(b_1,\ldots,b_m)$, where
$m$ is the unique $j \in [1,n]$ such that $b'_i = \True$ for all 
$i \in [1,j-1]$ and $b'_j = \False$ if such an $m$ exists and
$g_{2n}(b_1,b'_1,\ldots,b_n,b'_n) = \False$ otherwise; and
$\smash{\funct{g_{2n+1}}{\Bool^{2n+1}}{\Bool}}$ defined by 
$g_{2n+1}(b_1,b'_1,\ldots,b_n,b'_n,b) =
 g_{2n}(b_1,b'_1,\ldots,b_n,b'_n)$.
\qed
\end{proof}

The following result is a corollary of Theorem~\ref{theorem-projective}
and the definitions of \PLIS\ and $\pnuc{\ISbr\hsp{-.1}}{\poly}$.
\begin{corollary}
\label{corollary-projective}
Let $\indfam{f_n}{n \in \Nat} \in \PLIS$.
Then there exists a
$\indfam{g_n}{n \in \Nat} \in \pnuc{\ISbr\hsp{-.1}}{\poly}$ such that
$\indfam{f_n}{n \in \Nat} \plred{\ISbr}{\lin}
 \indfam{g_n}{n \in \Nat}$.
\end{corollary}

\section{Concluding Remarks}
\label{sect-concl}

We have presented an approach to non-uniform complexity which is based 
on the simple idea that each Boolean function can be computed by a 
single-pass instruction sequence that contains only instructions to read 
and write the contents of Boolean registers, forward jump instructions, 
and a termination instruction.

We have answered various questions that arise from this approach, but 
many open questions remain.
We mention:
\begin{itemize}
\item
We do not know whether Theorem~\ref{theorem-hierarchy-NPLIS} can be 
sharpened.
In particular, it is an open question whether, for each $k \in \Nat$, 
$\ndnuc{\ISbr}{O(n^k)} \subset \ndnuc{\ISbr}{O(n^{k+1})}$.
\item
We know little about complexity classes $\nuc{\IS}{\FN}$ where 
$\IS \subset \ISbr$.
In particular, it is an open question whether:
\begin{itemize}
\item 
$\nuc{\ISbrna\hsp{-.1}}{\poly} \subset \PLIS$;
\item
for each $l \in \Nat$, 
$\nuc{\RISbr{0}{l}}{\poly} \subset \nuc{\ISbrna\hsp{-.1}}{\poly}$;
\item
for each $k \in \Nat$, 
$\nuc{\ISbrna}{O(n^k)} \subset \nuc{\ISbr}{O(n^k)}$;
\item
for each $k,l \in \Nat$, 
$\nuc{\RISbr{0}{l}}{O(n^k)} \subset \nuc{\ISbrna\hsp{-.1}}{O(n^k)}$.
\end{itemize}
\item
Likewise, we know little about complexity classes $\ndnuc{\IS}{\FN}$ 
where $\IS \subset \ISbr$.
It is also an open question whether: 
\begin{itemize}
\item 
$\ndnuc{\ISbrna\hsp{-.1}}{\poly} \subset \NPLIS$;
\item
for each $l \in \Nat$, 
$\ndnuc{\RISbr{0}{l}}{\poly} \subset \ndnuc{\ISbrna\hsp{-.1}}{\poly}$;
\item
for each $k \in \Nat$, 
$\ndnuc{\ISbrna}{O(n^k)} \subset \ndnuc{\ISbr}{O(n^k)}$;
\item
for each $k,l \in \Nat$, 
$\ndnuc{\RISbr{0}{l}}{O(n^k)} \subset \ndnuc{\ISbrna\hsp{-.1}}{O(n^k)}$.
\end{itemize}
\item
We also know little about the connections between complexity classes 
$\nuc{\IS}{\FN}$ and $\ndnuc{\IS}{\FN}$ with $\IS \subset \ISbr$ and 
classical complexity classes.
In particular, it is an open question whether there are classical 
complexity classes that coincide with the complexity classes
$\nuc{\RISbr{0}{l}}{\poly}$, $\ndnuc{\ISbrna\hsp{-.1}}{\poly}$, and
$\ndnuc{\RISbr{0}{l}}{\poly}$.
\end{itemize}
There are not yet indications that the above-mentioned open questions 
concerning proper inclusions of complexity classes $\nuc{\IS}{\FN}$ and 
$\ndnuc{\IS}{\FN}$ with $\IS \subset \ISbr$ are interdependent.

\pagebreak[2]
It is easy to see that $\nuc{\ISbrna}{\poly}$ coincides with the 
classical complexity class $\mathrm{L/poly}$.
It is well-known that, for all $\funct{f}{\seqof{\Bool}}{\Bool}$,
$f \in \mathrm{L/poly}$ iff $f$ has polyno\-mial-size branching programs
(see e.g.\ Theorem~4.53 in~\cite{Thi00a}).%
\footnote
{$\mathrm{L}$ is the class of all $\funct{f}{\seqof{\Bool}}{\Bool}$ that
are logarithmic-space computable, see e.g.\ Chapter~4 of~\cite{AB09a}.}
Let $\indfam{f_n}{n \in \Nat} \in \nuc{\ISbrna\hsp{-.1}}{\poly}$.
Then, for all $n \in \Nat$, the thread produced by the instruction 
sequence that computes $f_n$ is in essence a branching program and its 
size is polynomially bounded in $n$.
As a consequence of this, $\nuc{\ISbrna}{\poly}$ coincides with 
$\mathrm{L/poly}$.

The approaches to computational complexity based on loop
programs~\cite{MR67a}, straight-line programs~\cite{GLF77a}, and
branching programs~\cite{BDFP86a} appear to be the closest related to
the approach followed in this paper.

The notion of loop program is far from abstract or general: a loop
program consists of assignment statements and possibly nested loop
statements of a special kind.
Loop programs are nevertheless closer to instruction sequences than
Turing machines or Boolean circuits.
After a long period of little interest, there is currently a revival of
interest in the approach to issues relating to non-uniform
computational complexity based on loop programs (see
e.g.~\cite{BJK08a,KN04a,NW06a}).
The notion of loop program used in recent work on computational 
complexity is usually more general than the one originally used.

The notion of straight-line program is relatively close to the notion
of single-pass instruction sequence: a straight-line program is a
sequence of steps, where in each step a language is generated by
selecting an element from an alphabet or by taking the union,
intersection or concatenation of languages generated in previous steps.
In other words, straight-line programs can be looked upon as single-pass 
instruction sequences with special basic instructions, and without test 
and jump instructions.
To our knowledge, the notion of straight-line program is only used in 
the work presented in~\cite{BDG85a,GLF77a}.

The notion of branching program is actually a generalization of the
notion of decision tree from trees to graphs, so the term branching
program seems rather far-fetched.
However, branching programs are in essence threads, i.e.\ the objects 
that we use to represent the behaviours produced by instruction 
sequences under execution.
Branching programs are related to non-uniform space complexity like
Boolean circuits are related to non-uniform time complexity.
Like the notion of Boolean circuit, the notion of branching program
looks to be lasting in complexity theory 
(see e.g.~\cite{Thi00a,Weg00a}).

The complexity class \NPLIS\ can alternatively be defined in the same 
style as \PLIS\ in a setting that allows instruction sequence splitting.
In~\cite{BM08g}, we introduce an extension of \PGA\ that allows 
single-pass instruction sequence splitting and an extension of \BTA\ 
with a behavioural counterpart of instruction sequence splitting that is 
reminiscent of thread forking, and define \NPLIS\ in this alternative 
way.

\bibliographystyle{splncs03}
\bibliography{IS}

\end{document}